\documentclass{birkjour}
\usepackage{amsmath}
\usepackage{amssymb}
\usepackage{amsfonts}
\usepackage{graphicx}
\usepackage{color}
\usepackage{times}
\usepackage{datetime}
\usepackage{framed}
\usepackage{hyperref}
\usepackage{datetime}
\usepackage{framed} % for pdf, bitmapped graphics files
\usepackage{graphicx} % for pdf, bitmapped graphics files
\usepackage{amsmath} % assumes amsmath package installed
\usepackage{amssymb}  % assumes amsmath package installed
 \newtheorem{thm}{Theorem}[section]
 
 \newtheorem{lem}[thm]{Lemma}
 
 \theoremstyle{definition}
 
 \theoremstyle{remark}

 \numberwithin{equation}{section}

\def\proof{\paragraph*{\bf Proof. }}

\def\<{\leqslant}           % nice less than or equal to sign
\def\>{\geqslant}           % nice larger than or equal to sign
\def\div{\mathrm{div}}         % divergence
\def\d{\partial}
\def\wh{\widehat}
\def\wt{\widetilde}

\def\Argmin{\mathop{\mathrm{Argmin}}}
\def\Argmax{\mathop{\mathrm{Argmax}}}

\def\cH{\mathcal{H}}   % Hardy space
\def\mZ{\mathbb{Z}}    % set of integers
\def\mR{\mathbb{R}}    % real line
\def\mC{\mathbb{C}}    % complex plane
\def\Tr{\mathrm{Tr}}       % matrix trace
\def\rT{\mathrm{T}}        % matrix transpose
\def\rB{\mathrm{B}}        % matrix transpose
\def\bE{\mathbf{E}}    % expectation

\def\[[[{[\![\![}
\def\]]]{]\!]\!]}
\def\bra{{\langle}}
\def\ket{{\rangle}}
\def\Bra{\Big\langle}
\def\Ket{\Big\rangle}
\def\re{\mathrm{e}}        % number e
\def\rd{\mathrm{d}}        % differential
\def\cL{\mathcal{L}}
\def\bD{\mathbf{D}}
\def\bF{\mathbf{F}}

\def\x{\times}

\def\fB{\mathfrak{B}}

\def\fS{{\mathfrak S}}

\def\mP{\mathbb{P}}

\def\bh{\mathbf{h}}

\def\cF{\mathcal{F}}

\def\fX{\mathfrak{X}}

\def\cG{\mathcal{G}}

\def\cA{\mathcal{A}}
\def\cB{\mathcal{B}}

\def\cov{\mathbf{cov}}

\def\cT{{\mathcal T}}

\def\mS{\mathbb{S}}
\def\mT{\mathbb{T}}
\def\mZ{\mathbb{Z}}

\begin{document}
%-------------------------------------------------------------------------
% editorial commands: to be inserted by the editorial office
%
%\firstpage{1} \volume{228} \Copyrightyear{2004} \DOI{003-0001}
%
%
%\seriesextra{Just an add-on}
%\seriesextraline{This is the Concrete Title of this Book\br H.E. R and S.T.C. W, Eds.}
%
% for journals:
%
%\firstpage{1}
%\issuenumber{1}
%\Volumeandyear{1 (2004)}
%\Copyrightyear{2004}
%\DOI{003-xxxx-y}
%\Signet
%\commby{inhouse}
%\submitted{March 14, 2003}
%\received{March 16, 2000}
%\revised{June 1, 2000}
%\accepted{July 22, 2000}
%
%
%
%---------------------------------------------------------------------------
%Insert here the title, affiliations and abstract:

\title[Conditioning, Entropy and Equilibrium in Stochastic Hamiltonian Systems]
 {\Large Position-Momentum Conditioning, Relative Entropy Decomposition  and Convergence to Equilibrium in Stochastic Hamiltonian Systems}

%----------Author 1
\author[Igor G. Vladimirov]{Igor G. Vladimirov}%: \today, \currenttime}

\address{%
Australian National University,
 ACT 2601,
  Canberra,
   Australia}

\email{igor.g.vladimirov@gmail.com}

\thanks{This work is supported by the Australian Research Council  grants DP210101938 and DP200102945.}

%----------classification, keywords, date
 % Computational methods for ergodic theory (approximation of invariant measures, computation of Lyapunov exponents, entropy


%\date{16 August 2017}

\begin{abstract}
This paper is concerned with a class of multivariable stochastic Hamiltonian systems
whose
generalised position is related by an ordinary differential equation to the momentum %which is
governed by an Ito stochastic differential equation.   The latter is driven by a standard Wiener process and involves both conservative and viscous damping forces. With the mass,  diffusion and damping matrices being position-dependent,   the resulting nonlinear model of Langevin dynamics describes dissipative mechanical systems (possibly with rotational degrees of freedom) or their electromechanical analogues subject to external random forcing.
We study the time evolution of the joint position-momentum probability distribution for the system and its convergence to equilibrium  by decomposing the Fokker-Planck-Kolmogorov equation (FPKE) and the Kullback-Leibler  relative entropy with respect to the invariant measure into those for the position distribution and the momentum distribution conditioned on the position. This decomposition reveals a manifestation of the Barbashin-Krasovskii-LaSalle principle
and higher-order dissipation inequalities for the relative entropy as a Lyapunov functional  for the FPKE.
\end{abstract}

%%% ----------------------------------------------------------------------
\maketitle
%%% ----------------------------------------------------------------------
%\tableofcontents

\section{Introduction}\label{sec:intro}

The dynamics of a physical system and its interaction with the environment are strongly influenced by the transformation of energy from one type to another (or energy exchange between different subsystems) and its dissipation into heat (through mechanical friction,  viscous damping, electrical resistance or other mechanisms). The energy balance relations take into account both conservative and nonconservative forces and   are incorporated in the structure of equations of motion  in Lagrangian and Hamiltonian mechanics which  applies to mechanical systems, electrical circuits and their electromechanical hybrids. 
Energy transfer and dissipation  are crucial in control-by-interconnection  for
port-Hamiltonian systems \cite{VJ_2014}.
The Hamiltonian, which describes the total energy, is a natural candidate for a Lyapunov function in analysis of stability of equilibria in dissipative systems.
However, in the case of a Langevin damping force \cite{Z_2001}, which depends linearly on the generalised velocity, the time derivative of the Hamiltonian along a phase trajectory vanishes whenever the velocity (and hence, the generalised momentum) does. This prevents the Hamiltonian from being a strict Lyapunov function and requires subtler considerations, such as application of the Barbashin-Krasovskii-LaSalle (BKL) invariance principle \cite{BK_1952,L_1960}, 
which allows the asymptotic stability in this case  to be established  through verifying that the time derivative of the Hamiltonian vanishes forever only at equilibrium (that is, at a stationary point of the potential energy and zero momentum). The particular form of  the BKL principle for damped Hamiltonian systems is dictated by the symplectic structure, associated with the splitting of the dynamic variables into positions and momenta and the total energy into the potential and kinetic parts.

In the presence of external random forcing,  the resulting stochastic  Hamiltonian system  with dissipation has more complicated equilibria which are organised as invariant measures (rather than points) in the phase space. Such a system  is governed by an ODE for the position (relating the velocity  to the momentum through a mass matrix) and an Ito  stochastic differential equation (SDE) for the momentum. The latter is driven by a standard Wiener process \cite{KS_1991} (weighted by the square root of a diffusion matrix),  and its drift term involves a conservative force (the negative of the gradient of the potential energy) along with a Langevin viscous  damping force (specified by a damping matrix).
The mass,  diffusion and damping matrices can be position-dependent, thus making the Langevin dynamics applicable to dissipative mechanical systems   influenced by a random environment (such as flexible structures with rotational degrees of freedom,  where the moments of inertia, dissipation and exposure to external forces depend on the current spatial configuration of the system).
This class of systems (including their electromechanical analogues) is used, for example, in modelling molecular dynamics \cite{H_1986,SK_2020,T_2010},
turbulent  fluid-structure interaction (see \cite{I_2019} and references therein) and electrical circuits with thermal noise \cite{R_1974}.

The invariant measure of a dissipative stochastic Hamiltonian system is a steady-state solution of the Fokker-Planck-Kolmogorov equation (FPKE) \cite{BKRS_2015,R_1996,S_2008} which,
in this case,
describes the time evolution of the joint probability density function (PDF) of the position and momentum variables forming a diffusion process in the phase space. Under a multivariate version of the Einstein relation \cite{K_1966} between the damping and diffusion matrices, the Maxwell-Boltzmann PDF, which originates from the minimum Helmholtz  free energy principle of equilibrium statistical mechanics \cite{R_1978} and is completely specified by an inverse  temperature parameter and the Hamiltonian, is an invariant position-momentum PDF. Since the corresponding invariant PDF for the position variables achieves its maxima and is concentrated (asymptotically, as the temperature goes to zero) at the minima of the potential energy, the stochastic Hamiltonian models extend the Smoluchowski  SDEs (employed as randomised algorithms for minimising the potential energy) by taking into account the inertial effects. A similar extension is used in the ``heavy-ball'' method \cite{P_1964}, which adds inertia to the gradient descent algorithms, making them reminiscent of  the physical system dynamics.

It is customary to describe the convergence of the system to its thermal equilibrium (with an external  heat bath) in terms of the Kullback-Leibler   relative entropy \cite{CT_2006}, which, in the stochastic Hamiltonian setting,  is applied to the position-momentum PDF and is considered with respect to the invariant Maxwell-Boltzmann PDF by using an appropriate PDF ratio.
In addition to being proportional to the deviation of the Helmholtz free energy from its minimum value, this entropy has close links with information theory \cite{G_2008,ME_1981}. Furthermore, over the course of evolution of the dissipative stochastic Hamiltonian system, the relative entropy does not increase in time, which makes it a natural candidate for a Lyapunov functional on the infinite-dimensional manifold of position-momentum PDFs rather than the finite-dimensional phase space of the system. However, similarly to the Hamiltonian in the deterministic dissipative case, the time derivative of the relative entropy can vanish at some moments of time before the equilibrium is reached, and hence, the entropy cannot be directly used as a strict Lyapunov functional for the convergence of the stochastic Hamiltonian system to its equilibrium measure. A detailed study of this issue is the main thrust of the present paper.

More precisely, we consider the class of multivariable dissipative stochastic Hamiltonian systems satisfying the damping-diffusion relation mentioned above. For such a system, the position-momentum PDF is factorised into a position PDF and a momentum PDF conditioned on the position, and a similar factorisation holds for the invariant PDF. Accordingly, the relative entropy of the position-momentum distribution with respect to the invariant  measure is decomposed into the sum of the position entropy and the conditional momentum entropy. We represent the FPKE for the position-momentum PDF  in the form of two coupled PDEs for the position PDF and the conditional momentum PDF. These PDEs lead to  higher-order dissipation relations for the position and conditional momentum entropies, along with the property that the time derivative of the position-momentum entropy vanishes whenever the conditional momentum PDF coincides with its equilibrium counterpart everywhere in the phase space. We show that, unless the total equilibrium is already reached (that is, unless the position PDF is also at equilibrium), any such ``entropy dissipation break'' is followed by an increase in the conditional momentum entropy over a sufficiently short interval of time, which is  accompanied by a decrease in the position entropy in such a way that the position-momentum relative entropy, as a function of time, experiences a stationary point of inflection. As a result, the entropy break instants (at which the total relative entropy has zero time derivative  before the system reaches its equilibrium) are isolated and form a countable set, thus making the relative entropy a strictly decreasing function of pre-equilibrium time. The structure of this set (and the ``waterbed'' effect, by which the increase in one of the entropy components causes a decrease in the other as if there were an entropy flow between them)  is a manifestation of the BKL principle in application to the relative entropy as a Lyapunov functional  for the FPKE. In the stochastic Hamiltonian setting, this principle reveals itself through the position-momentum conditioning, so that the position PDF and the position entropy play a role similar to the position and potential energy, respectively, while the conditional momentum PDF and the conditional momentum entropy correspond to the momentum and the kinetic energy. This correspondence between the total energy dissipation in the deterministic case and the position-momentum relative entropy dissipation in the stochastic case is the principal message of the present study, which does not consider the rate of this dissipation in the long run discussed, for example, in \cite{HN_2004,T_2002}. Nevertheless, we also outline the role which the position-momentum conditioning can play for spectral analysis of the PDF dynamics in the  stochastic Hamiltonian setting through linearised dynamics of the position and conditional momentum logarithmic PDF ratios near the equilibrium.

The paper is organised as follows.
Section~\ref{sec:sys} specifies the class of stochastic Hamiltonian systems under consideration.
Section~\ref{sec:PDF} discusses the PDF dynamics for the position-momentum state process of the system.
Section~\ref{sec:inv} provides a multivariate damping-diffusion relation and
the invariant PDF corresponding to the postulate of equilibrium statistical mechanics.
Section~\ref{sec:posmoment} describes the decomposition of the position-momentum relative entropy with respect to the invariant PDF into the position and conditional momentum entropies.
Section~\ref{sec:relent} establishes dissipation relations for the position-momentum entropy in terms of its components and studies their behaviour in the vicinity of the entropy breaks.
Section~\ref{sec:lin} outlines a linearisation of the position and conditional momentum PDF dynamics. 
Section~\ref{sec:conc} makes concluding remarks.

\section{Stochastic Hamiltonian Systems with Damping}
\label{sec:sys}

We consider a class of stochastic Hamiltonian systems with $n$ degrees of freedom and phase space $S\x \mR^n$. For simplicity,  the position space $S$ is assumed to be either  $\mR^n$, the $n$-dimensional  torus $\mT^n$,  or the product $\mR^{n-r}\x \mT^r$ of such sets, where $\mT:= \mR / (2\pi \mZ)$  corresponds to a rotational degree of freedom and is identified with the interval $[0,2\pi)$ (so that functions on $\mT$ can be viewed as $2\pi$-periodic functions on $\mR$).
In each of these cases, $S$ is endowed with the $n$-dimensional Lebesgue measure $\lambda_n$.
The position of the system is specified by a vector  $q:= (q_k)_{1\< k\< n} \in S$ of generalised coordinates whose time derivatives form the generalised velocity $\dot{q} := (\dot{q}_k)_{1\< k\< n} \in \mR^n$ (vectors are assumed to be organised as columns). The kinetic energy
\begin{equation}
 \label{T}
    T(x)
    :=
    \frac{1}{2} \|p\|_{M(q)^{-1}}^2
    =
    \frac{1}{2}
    \|\dot{q}\|_{M(q)}^2
\end{equation}
of the system is a position-dependent quadratic form
in the generalised momentum vector
\begin{equation}
\label{pqdot}
    p
    :=
    \d_{\dot{q}}
    T(x)
    =
    M(q)\dot{q},
\end{equation}
where
\begin{equation}
\label{xqp}
    x
    :=
    (x_k)_{1\< k \< 2n}
    :=
    \begin{bmatrix}
      q\\
      p
    \end{bmatrix}
\end{equation}
is the position-momentum vector. For any $q\in S$, the generalised mass matrix
$M(q)$ is a real positive definite symmetric matrix of order $n$ (we denote the set of such matrices by $\mP_n$), which, in the case of rotational degrees of freedom, represents the tensor of inertia. Also, $\|v\|_N:= |\sqrt{N}v| = \sqrt{v^{\rT}N v}$ in (\ref{T}), (\ref{pqdot}) is a weighted Euclidean (semi-) norm of a real vector $v$, specified by a  real positive (semi-) definite symmetric  matrix $N$, with $|\cdot|$ the standard Euclidean norm.  The system Hamiltonian  $H: S \x \mR^n \to \mR$ is the sum of the kinetic energy (\ref{T}) and
the potential energy $V: S \to \mR$:
\begin{equation}
\label{HTV}
    H(x) := T(x) + V(q),
    \qquad
    q \in S,\
    p \in \mR^n.
\end{equation}
The functions $V$, $M$ are assumed to be twice continuously differentiable. 
The position $Q(t)$ and the momentum $P(t)$ of the system (depending on time $t\>0$)  are random processes with multivariable Langevin dynamics \cite{Z_2001} governed by an ODE
\begin{equation}
\label{SH1}
    \dot{Q}
    =
    M(Q)^{-1}P
    =
        \d_p H(X)
\end{equation}
and an Ito SDE \cite{KS_1991}
\begin{equation}
\label{SH2}
    \rd P
     =
    -(\d_q H(X) + F(Q)\dot{Q})\rd t
    +
    \sqrt{D(Q)}\rd W,
\end{equation}
driven by a  standard $\mR^n$-valued  Wiener process $W$, independent of the initial position and momentum $Q(0)$, $P(0)$.  The function $F: S\to \mP_n$   
specifies the Langevin viscous damping force $-F(Q)\dot{Q}$, and 
$D:S \to \mP_n$ describes the diffusion matrix for the SDE (\ref{SH2}). The $S\x \mR^n$-valued state process
\begin{equation}
\label{XQP}
    X
    :=
    \begin{bmatrix}
        Q\\
        P
    \end{bmatrix}
\end{equation}
satisfies the augmented SDE
\begin{align}
\nonumber
    \rd X
    & =
    \left(
        J H'(X)
        -
        \begin{bmatrix}
          0 \\
          F(Q)
        \end{bmatrix}
        \dot{Q}
    \right)
        \rd  t
        +
        \begin{bmatrix}
          0 \\
          \sqrt{D(Q)}
        \end{bmatrix}
        \rd W\\
\label{dX}
    & =
    \left(
        J
        -
        \begin{bmatrix}
          0 & 0 \\
          0 & F(Q)
        \end{bmatrix}
    \right)
    H'(X)
        \rd  t
        +
        \begin{bmatrix}
          0 \\
          \sqrt{D(Q)}
        \end{bmatrix}
        \rd W,
\end{align}
obtained by combining (\ref{SH1}), (\ref{SH2}), where
$(\cdot)'$ is the gradient of a function with respect to all its variables, so that
\begin{equation*}
\label{H'}
    H'
    =
    \d_x H
    =
    \begin{bmatrix}
        \d_q H \\
        \d_p H
    \end{bmatrix}
\end{equation*}
consists of the gradients of the Hamiltonian over the positions and momenta in (\ref{xqp}), with
\begin{equation}
\label{dHdq}
    \d_q H
     =
    V'(q)
    -
    \frac{1}{2}
    (
        p^{\rT}M_k(q)p
    )_{1\< k\< n}
    =
    V'(q)
    -
    \frac{1}{2}
    (
        \dot{q}^\rT (\d_{q_k}M) \dot{q}
    )_{1\< k\< n}
\end{equation}
in view of (\ref{T}), (\ref{HTV}).
Here, $-V' = -\d_qV$ is the potential force field, while the additional ``centrifugal'' term in (\ref{dHdq}), which depends quadratically  on the velocity $\dot{q}$,   originates from the dependence of the mass matrix $M$ on $q$:
\begin{equation}
\label{muk}
  M_k(q)
  :=
  -\d_{q_k}(M(q)^{-1})
  =
  M(q)^{-1}
  \d_{q_k}M(q)
  M(q)^{-1},
\end{equation}
so that the resulting maps $M_1, \ldots, M_n: S\to \mS_n$  take values in the subspace $\mS_n$ of  real symmetric matrices of order $n$. 
In (\ref{dX}), use is also made of the symplectic structure matrix
\begin{equation}
\label{J}
    J:=
    \begin{bmatrix}
        0 & I_n \\
        -I_n & 0 \\
    \end{bmatrix},
\end{equation}
where 
$I_n$ is the identity  matrix of order $n$. In particular, constant mass, damping and diffusion matrices and a quadratic potential energy lead to linear stochastic Hamiltonian systems (robust stabilization and control problems with mean square optimality criteria for such systems and their interconnections are considered, for example,  in \cite{VP_2018,VP_2020}).

For a twice continuously differentiable test function $\varphi \in C^2(S \x \mR^n, \mR)$, with the gradient vector and Hessian matrix
$$
    \varphi'
    =
    \d_x\varphi
    =
    {\begin{bmatrix}
        \d_q \varphi \\
        \d_p \varphi
    \end{bmatrix}},
    \qquad
    \varphi''
    =
    \d_x^2\varphi
    =
    {\begin{bmatrix}
        \d_q^2 \varphi & \d_p\d_q \varphi  \\
        \d_q\d_p \varphi & \d_p^2 \varphi
    \end{bmatrix}}
$$
(where $\d_q\d_p \varphi:= (\d_{q_k}\d_{p_j} \varphi)_{1\< j,k\< n} = (\d_p\d_q \varphi)^\rT$),
the infinitesimal generator $\cL$ of the diffusion process $X$ in (\ref{dX}) acts as
\begin{align}
\nonumber
  \cL(\varphi)
  & =
  -
  H'^\rT
    \left(
        J
        +
        \begin{bmatrix}
          0 & 0 \\
          0 & F
        \end{bmatrix}
    \right)
    \varphi'
    +
    \frac{1}{2}
    \Bra
        \begin{bmatrix}
          0 & 0 \\
          0 & D
        \end{bmatrix},
                \varphi''
    \Ket\\
\label{cL}
    & =
    \underbrace{-\{H,\varphi\}}_{\{\varphi,H\}}
    -
    p^\rT M^{-1}F\d_p\varphi
    +
    \frac{1}{2}
    \bra
        D,
        \d_p^2 \varphi
    \ket.
\end{align}
Here, $\bra K, N\ket:= \Tr (K^\rT N)$ is the Frobenius inner product of real matrices \cite{HJ_2007}, and   use is made of the Poisson bracket  \cite{A_1989}
\begin{equation}
\label{Poiss}
    \{u, v\}
    :=
    u'^{\rT} J v'
    =
    \d_q u^{\rT}\d_p v - \d_p u^{\rT}\d_q v
    =
    -\{v,u\}
\end{equation}
for continuously differentiable functions $u, v \in C^1(S \x \mR^n, \mR)$ on the phase space,    with the rightmost equality in (\ref{Poiss}) 
following from the antisymmetry $J=-J^\rT$ in (\ref{J}).

Since the system involves damping and is subject to external forcing, the Hamiltonian $H(X)$  is no longer an  integral of motion. As a twice continuously differentiable function evaluated at the process (\ref{XQP}), it  has the stochastic differential
\begin{align}
\nonumber
    \rd H
    & =
    H'^{\rT}\rd X
    +
    \frac{1}{2} \bra \d_p^2 H, D\ket\rd t\\
\nonumber
    & =
    H'^{\rT}
        \left(
        J
        -
        \begin{bmatrix}
          0 & 0 \\
          0 & F
        \end{bmatrix}
    \right)
    H'
        \rd  t
        +
        \d_p H(X)^{\rT}
        \sqrt{D}\rd W
        +
        \frac{1}{2} \bra \d_p^2 H, D\ket\rd t\\
\label{dH}
        & =
        \underbrace{\Big(\frac{1}{2}\bra M^{-1}, D\ket-\|\dot{Q}\|_F^2\Big)}_{\cL(H)}\rd t
        +
        \dot{Q}^\rT
        \sqrt{D}\rd W,
\end{align}
obtained from (\ref{dX}) by using the Ito lemma \cite{KS_1991}. Evaluation of the generator (\ref{cL}) at the Hamiltonian in the drift term of (\ref{dH}) employs the identity
\begin{equation}
\label{HH0}
    H'^\rT J H' = \{H,H\} = 0
\end{equation}
(which secures the preservation of $H$ in the case of isolated conservative system dynamics due to the antisymmetry 
of the Poisson bracket (\ref{Poiss})).
The functions $F$, $M$, $D$ in (\ref{dH}) are evaluated at the position process $Q$, and $\d_p^2 H = M^{-1}$ is the Hessian matrix of the Hamiltonian with respect to the momenta in view of (\ref{T}), (\ref{HTV}). The Ito correction term $\frac{1}{2}\bra M^{-1}, D\ket \> 0$ contributes to the energy gain for the system, while  $-\|\dot{Q}\|_F^2 \< 0$ has an opposite effect of energy dissipation due to damping.
Since the martingale part $\dot{Q}^{\rT}
        \sqrt{D}\rd W$ of the SDE (\ref{dH}) does not influence the time derivative of the mean value of $H$, then
\begin{align}
\nonumber
    (\bE H)^{^\centerdot}
    & =
    \bE \cL(H)
    = \bE \Big(\frac{1}{2}\bra M^{-1}, D\ket-\|\dot{Q}\|_F^2\Big)\\
\nonumber
    & =
    \bE \Big(\frac{1}{2}\bra M^{-1}, D\ket-\bE(\|P\|_{M^{-1}FM^{-1}}^2\mid Q)\Big)\\
\label{EHdot}
    &  =
    \bE \Bra M^{-1},  \frac{1}{2}D-\bE(PP^{\rT}\mid Q)M^{-1} F\Ket.
\end{align}
Here, $\bE(\cdot)$ is expectation, and  the tower property \cite{S_1996}  of iterated conditional expectations $\bE(\cdot \mid \cdot)$ is used along with the matrix of conditional second-order moments of $P$,  given $Q$:
\begin{equation}
\label{EPPQ}
    \bE(PP^{\rT}\mid Q)
    =
    \bE(P\mid Q)\bE(P\mid Q)^\rT + \cov(P\mid Q),
\end{equation}
where $\bE(P\mid Q)$ and  $\cov(P\mid Q)$ are the corresponding    conditional mean vector and covariance matrix.  A similar equation for the mean value of the kinetic energy (\ref{T}) is obtained from (\ref{HTV}), (\ref{EHdot}) as
\begin{align}
\nonumber
    (\bE T)^{^\centerdot}
    & =
    (\bE H)^{^\centerdot} - (\bE V)^{^\centerdot}
    =
    (\bE H)^{^\centerdot} - \bE (V'^\rT \dot{Q})\\
\nonumber
    & =
    (\bE H)^{^\centerdot} - \bE (V'^\rT M^{-1}\bE(P\mid Q))
    =
    (\bE H)^{^\centerdot} - \bE \bra M^{-1}, \bE(P\mid Q) V'^\rT\ket\\
\label{ETdot}
    &  =
    \bE \Bra M^{-1},  \frac{1}{2}D-\bE(PP^{\rT}\mid Q)M^{-1} F- \bE(P\mid Q) V'^\rT\Ket.
\end{align}

In the limiting case of a deterministic Hamiltonian system (with dissipation) isolated from the external random force, when $ D=0$ and the SDEs (\ref{SH2}), (\ref{dX}), (\ref{dH}) are ODEs,  the relations (\ref{EHdot}), (\ref{ETdot}) reduce to
\begin{equation}
\label{HVTdot}
    \dot{H} = -\|M^{-1}P\|_F^2 \< 0,
    \quad
    \dot{V} = V'^\rT M^{-1}P,
    \quad
    \dot{T} = -\|M^{-1}P\|_F^2 - \dot{V}.
\end{equation}
In this case, $\dot{H} = 0$  whenever $P=0$, and hence, the Hamiltonian $H$ cannot be directly employed as a strict Lyapunov function for the convergence of the phase trajectories $X$ to the set
\begin{equation}
\label{Hmin}
    \Argmin_{x\in S \x \mR^n}
    H(x)
    =
    \Argmin_{q\in S}V(q)\x \{0\}
\end{equation}
(assuming that the potential energy $V$ achieves its minimum in the position space  $S$). Nevertheless, at any such moment of time (when $P=0$), $H$ has zero second-order  time derivative $\ddot{H}=0$,  and, since the Langevin damping force in (\ref{SH2}) with $D=0$ and the centrifugal term in (\ref{dHdq}) vanish, then $\dot{P} = -V'$, thus leading,    in view of (\ref{HVTdot}),   to  the negative third-order time derivative
\begin{equation}
\label{H...}
    \dddot{H} = -2\|M^{-1} V'\|_F^2 <0
\end{equation}
unless the current position $Q$ of the system is a stationary point of $V$. Hence, any instant $t_0>0$, such that  $V'(Q(t_0))\ne 0$ and $P(t_0)=0$,  is a stationary point of inflection for the Hamiltonian $H$ as a function of time along the trajectory of the system.  Furthermore,     in view of (\ref{HVTdot}), the potential and kinetic energies experience strict local maximum and minimum, respectively,  at any such moment of time $t_0$  since
\begin{align}
\label{V0}
    \dot{V}
    & = 0,
    \qquad
    \ddot{V} = -\|V'\|_{M^{-1}}^2 < 0,\\
\label{T0}
    T & = 0,
    \qquad
    \dot{T} = 0,
    \qquad
    \ddot{T} = -  \ddot{V} > 0,
\end{align}
which can be interpreted as a local  exchange between these two types of energy (with the energy  dissipation vanishing together with the damping force). This is understood   at the level of the first and second-order terms of the Taylor series expansions over a short time interval containing $t_0$. 
Therefore, if $V$ has a unique stationary point $q_* \in S$ (in which case,  $\Argmin_{q\in S}V(q) = \{q_*\}$ in (\ref{Hmin})), then the Hamiltonian is strictly decreasing in time  until the deterministic system  reaches the unique equilibrium point  $(q_*, 0)$ in the phase space.

Now, returning to the stochastic setting with $D\succ 0$ (where the equilibrium is organised as an invariant measure rather than a point in the  phase space), we note that
in comparison with the energy balance equations (\ref{dH}), (\ref{EHdot}), (\ref{ETdot}), which involve the first and second-order conditional  moments of the momentum,   a more detailed information on statistical properties of the system   is provided by the probability distribution of its state process.

\section{Position-Momentum PDF Dynamics and Conditioning}
\label{sec:PDF}

We assume that for any time $t\> 0$,  the probability distribution of the state vector $X(t)$ in (\ref{XQP}) is absolutely continuous with density $f(t,\cdot): S\x \mR^n \to \mR_+$ (with respect to the $2n$-dimensional  Lebesgue measure $\lambda_n\x \lambda_n = \lambda_{2n}$ on the phase space). It is also assumed  that $f$ is continuously differentiable in time and twice continuously differentiable in space. These assumptions can be justified 
by using the parabolic H\"{o}rmander condition \cite{H_1967,S_2008},  whose discussion is beyond the scope of this paper. 
We will now specify the FPKE for the PDF $f$ (the lemma below is given with a proof only for completeness and in order to provide additional notation for what follows).
\begin{lem}
\label{lem:FPKE}
The position-momentum PDF $f$ for the stochastic Hamiltonian system (\ref{SH1}), (\ref{SH2}), which is assumed to satisfy $f \in C^{1,2}(\mR_+\x S\x\mR^n, \mR_+)$, is governed by the FPKE
\begin{equation}
\label{FPKE}
    \d_t f
    =
    \{H,f\}
    +
    p^\rT M^{-1} F \d_pf
    +
    \bra
        F, M^{-1}
    \ket
    f
    +
    \frac{1}{2}
    \bra
        D, \d_p^2 f
    \ket.
\end{equation}
%\hfill$\square$
\end{lem}
\begin{proof}
In application to the diffusion process $X$ in  (\ref{dX}),  the FPKE for its PDF $f$ takes the form
\begin{equation}
\label{fdot}
    \d_t f
    =
    \cL^\dagger(f)=
    -\div
    \left(
    f
    \left(
        J
        -
        \begin{bmatrix}
          0 & 0 \\
          0 & F
        \end{bmatrix}
    \right)
    H'
    \right)
    +
    \frac{1}{2}
    \div^2
    \left(
        f
        \begin{bmatrix}
          0 & 0\\
          0 & D
        \end{bmatrix}
    \right),
\end{equation}
where $\cL^\dagger$ is the formal adjoint of the generator (\ref{cL}) (see, for example, \cite{S_2008}). Here, the divergence operator $\div(\cdot)$  in the phase space acts as
\begin{equation}
\label{divuv}
    \div
    \begin{bmatrix}
      u\\
      v
    \end{bmatrix}
    =
    \div_q u + \div_p v
\end{equation}
for $u,v \in C^1(S\x \mR^n, \mR^n)$ in accordance with (\ref{xqp}), where $\div_q(\cdot)$, $\div_p(\cdot)$ are the divergence operators over the position and momentum variables.   This operator is extended row-wise  to matrix-valued maps $w:=(w_{jk})_{1\< j\< r, 1\< k \< 2n} := \begin{bmatrix}
u & v
\end{bmatrix}
\in C^1(S\x \mR^n, \mR^{r\x 2n})$,   consisting of $u, v \in C^1(S\x \mR^n, \mR^{r\x n})$,     and yields an $\mR^r$-valued map
\begin{equation}
\label{divw}
    \div
    w
    :=
    \Big(
        \sum_{k=1}^{2n}\d_{x_k}w_{jk}
    \Big)_{1\< j\< r}
    =
    \div_q u + \div_p v,
\end{equation}
which coincides with the column-wise application of $\div (\cdot)$  if $w$ takes values in the subspace $\mS_{2n}$ (in which case, $r=2n$  and     $\div
    w
    =
    \big(
        \sum_{j=1}^{2n}\d_{x_j}w_{jk}
    \big)_{1\< k\< 2n}$ is $\mR^{2n}$-valued). 
Accordingly,
\begin{equation}
\label{divfJH}
    \div(fJH')
    =
    f'^\rT JH' + f\div (JH')
    =
    \{f,H\}
    =
    -
    \{H,f\},
\end{equation}
where the Poisson bracket (\ref{Poiss}) is used along with
the property (employed in Liouville's theorem 
\cite{A_1989}) that the Hamiltonian vector field $JH'$ is divergenceless ($\div(JH') = \div_q \d_p H - \div_p \d_q H= \Tr (\d_q \d_p H-\d_p \d_q H )=0$ in view of (\ref{divuv}) and the invariance of the trace under the matrix transposition).  Similar calculations lead to
\begin{align}
\nonumber
    \div
    \left(
        f
        \begin{bmatrix}
          0 & 0 \\
          0 & F
        \end{bmatrix}
        H'
    \right)
    & =
    \div
        \begin{bmatrix}
          0 \\
          fFM^{-1}p
        \end{bmatrix}
        =
    \div_p (fFM^{-1}p)\\
\nonumber
    & =
    p^\rT M^{-1} F  \d_p f
    +
    f\div_p(FM^{-1}p)\\
\label{divfFH}
    & =
    p^\rT M^{-1} F  \d_p f
    +
    \bra F, M^{-1}\ket f
\end{align}
in view of the mass and damping matrices $M$, $F$ being symmetric and depending only on the position vector $q$. Also,
$
    \div
    \left(
        f
        {\small\begin{bmatrix}
          0 & 0\\
          0 & D
        \end{bmatrix}}
    \right)
    =
        {\small\begin{bmatrix}
          0 & 0\\
          0 & D
        \end{bmatrix}}
        f'
        +
        f
        {\small\begin{bmatrix}
          0 \\
          \div_p D
        \end{bmatrix}}
        =
        {\small\begin{bmatrix}
          0 \\
          D \d_p f
        \end{bmatrix} }
$
in accordance with (\ref{divw}) and since the diffusion matrix $D$ depends  only on $q$,  and hence,
\begin{equation}
\label{div2}
    \div^2
    \left(
        f
        \begin{bmatrix}
          0 & 0\\
          0 & D
        \end{bmatrix}
    \right)
        =
        \div_p(D \d_p f)
        =
        \bra D, \d_p^2 f\ket.
\end{equation}
Substitution of (\ref{divuv})--(\ref{div2}) into (\ref{fdot}) allows the FPKE  to be represented in the form  (\ref{FPKE}). Note that the first terms on the right-hand sides of (\ref{cL}), (\ref{FPKE}) are related to each other by the formal skew self-adjointness property  of the Poisson bracket $\{H, \cdot\} = -\{H, \cdot\}^\dagger$.
\end{proof}
%%%%%%%%%%%%%%%%%%%%%%%%%%%%%%%%%%%%%%%%%%%%%%%%%%%%%%%%%%%%%%%%%%%%%%%%%%%%%%%%%%%%%%%

At any time $t\> 0$, the position $Q(t)$ of the system has a PDF $g(t,\cdot): S \to \mR_+$ obtained by integrating the position-momentum PDF $f(t,\cdot)$   over the momentum variables:
\begin{equation}
\label{gf}
    g(t,q) = \int_{\mR^n} f(t,q,p)\rd p.
\end{equation}
For what follows, it is assumed that the position PDF $g$ is positive everywhere:
\begin{equation}
\label{gpos}
  g(t,q)>0,
  \qquad
  t\> 0,
  \
  q \in S.
\end{equation}
Then the conditional PDF $h(\cdot \mid q, t): \mR^n\to \mR_+$ of the momentum $P(t)$, given $Q(t)=q$,  is provided by the ratio
\begin{equation}
\label{hfg}
  h(p\mid q, t)
  =
  \frac{f(t,q,p)}{g(t,q)},
\end{equation}
and the position-momentum PDF $f(t,\cdot)$ is factorised as
\begin{equation}
\label{fgh}
    f(t,q,p) = g(t,q)h(p\mid q,t),
    \qquad
    t \> 0,\ q\in S,\
    p \in \mR^n.
\end{equation}
In order to simplify the discussion of regularity issues in what follows, we assume that $g$ inherits from $f$ the smoothness properties and so also does the  conditional momentum  mean from (\ref{EPPQ}):
\begin{align}
\nonumber
  \gamma(t,q)
  & :=
  \bE(P(t)\mid Q(t)=q)
  =
  \int_{\mR^n}
   h(p\mid q,t) p
  \rd p\\
\label{EPQ}
  & =
  \frac{1}{g(t,q)}
  \int_{\mR^n}
   f(t,q,p) p
  \rd p,
  \qquad
  t\> 0,\
  q \in S,
\end{align}
which  is expressed here in terms of the conditional momentum PDF $h$ and its representation (\ref{hfg}).

\begin{lem}
\label{lem:gdot}
Under the assumptions of Lemma~\ref{lem:FPKE}, the position PDF $g$ in (\ref{gf})  satisfies a linear first-order PDE
\begin{equation}
\label{gdot}
  \d_t g = - \div_q(g M^{-1}\gamma),
\end{equation}
which involves the conditional momentum  mean from (\ref{EPQ}).%\hfill$\square$
\end{lem}
\begin{proof}
From (\ref{SH1}), it follows that
for any infinitely differentiable  function  $\varphi \in C^\infty(S,\mR)$ of bounded support (the latter condition is not needed in the case of the $n$-dimensional torus $S = \mT^n$),
\begin{align}
\nonumber
    (\bE \varphi(Q))^{^\centerdot}
    & =
    \bE (\varphi'(Q)^\rT \dot{Q})
    =
    \bE (\varphi'(Q)^\rT M(Q)^{-1}P)    \\
\nonumber
    & =
    \bE (\varphi'(Q)^\rT M(Q)^{-1} \bE (P\mid Q))
    =
    \int_S
    \varphi'(q)^\rT M(q)^{-1} \gamma(t,q)g(t,q)\rd q\\
\label{EphiQdot1}
    & =
    -
    \int_S
    \varphi(q)
    \div_q(g(t,q) M(q)^{-1} \gamma(t,q))\rd q.
\end{align}
Here, the tower property of conditional expectations is used along with (\ref{EPQ}), the identity $$\varphi'^\rT M^{-1} \gamma g = \div_q(\varphi M^{-1} \gamma g) - \varphi \div_q(M^{-1} \gamma g) $$ and the Gauss-Ostrogradsky theorem (applied to a subset of $S$ large enough to  contain the support of $\varphi$). On the other hand,
\begin{equation}
\label{EphiQdot2}
    (\bE \varphi(Q))^{^\centerdot}
    =
    \d_t
    \int_S
    \varphi(q) g(t,q)\rd q
    =
    \int_S
    \varphi(q)\d_t g(t,q)\rd q.
\end{equation}
Since the test function $\varphi$ is arbitrary,   a comparison of (\ref{EphiQdot1}) with (\ref{EphiQdot2}) leads to (\ref{gdot}). Alternatively, the PDE (\ref{gdot}) can also be established by integrating  both sides of the FPKE (\ref{FPKE}) over the momentum variables (according to (\ref{gf})) and using the divergence forms of (\ref{divfFH}), (\ref{div2}):
\begin{align}
\nonumber
    \d_t g
    & =
    \int_{\mR^n}
    \d_t f\rd p
    =
    \int_{\mR^n}
    \Big(
    \{H,f\}
    +
    \div_p
    \Big(
        f FM^{-1}p + \frac{1}{2} D\d_p f
    \Big)
    \Big)
    \rd p\\
\nonumber
    & = \int_{\mR^n}
    (
        \div_p(f\d_q H) - \div_q(f \d_p H)
    )
    \rd p
    =
    -\div_q
    \int_{\mR^n}
        f \d_p H
    \rd p\\
\label{gdot1}
    & =
    -\div_q
    \Big(
    g
    M^{-1}
    \int_{\mR^n}
        h p
    \rd p
    \Big)
    =
        -\div_q
    (
    g
    M^{-1}
    \gamma
    ),
\end{align}
where  the identity
\begin{equation}
\label{Poissdiv}
    \{u, v\}
    =
    \div_p(v\d_q u) - \div_q(v \d_p u)
\end{equation}
for the Poisson bracket (\ref{Poiss}) of functions $u \in C^2(S\x \mR^n, \mR)$ and $v \in C^1(S\x \mR^n, \mR)$ is applied to $u:= H$ and $v:= f(t,\cdot)$; cf. (\ref{divfJH}).  However, in accordance with the Gauss-Ostrogradsky  theorem, in order to justify the absence of contribution from the divergence terms $\div_p(\cdot)$  to the integrals in (\ref{gdot1}), an appropriately fast  decay of the corresponding vector fields at infinity  (more precisely, $o(|p|^{1-n})$, as $|p|\to +\infty$),  ensuring their flux decay, is needed as a sufficient condition.
\end{proof}

The techniques employed in (\ref{EphiQdot1})  are similar to those in A.Y.Klimenko's conditional moment closure   approach in fluid mechanics and combustion theory  \cite{KB_1999} and its adaptations to multiscale transport phenomena \cite{KA_2007,VK_2010}.
According to the PDE (\ref{gdot}), the evolution of the position PDF $g$ is driven by the conditional momentum average $\gamma$ in  (\ref{EPQ}) in the sense that if $\gamma=0$ everywhere in $\mR_+\x S$ then $g$ does not change in time.  Also, the PDE (\ref{gdot}) is not autonomous since its right-hand side  involves dependence (through $\gamma$) on the conditional momentum PDF $h$. The latter evolves as follows.

\begin{lem}
\label{lem:hdot}
Under the assumptions of Lemma~\ref{lem:FPKE} together with (\ref{gpos}), the conditional momentum PDF $h$ in (\ref{hfg})  satisfies the PDE
\begin{align}
\nonumber
    \d_t h
    =&
    \{H,h\}
    +
    p^\rT M^{-1} F \d_ph
    +
    \frac{1}{2}
    \bra
        D, \d_p^2 h
    \ket\\
\label{hdot}
     & +
     (    \bra
        F, M^{-1}
    \ket
    -\varpi^\rT M^{-1}\d_q \ln g
    + \div_q (M^{-1}\gamma))
     h,
\end{align}
where $\gamma$ is the conditional momentum  mean from (\ref{EPQ}). Here, the function $\varpi: \mR_+ \x S\x \mR^n \to \mR^n$ is given by
\begin{equation}
\label{varpi}
  \varpi(t,q,p):= p-\gamma(t,q),
  \qquad
  t\> 0, \
  q \in S,\
  p \in \mR^n
\end{equation}
and describes the conditionally centered momentum vector.%\hfill$\square$
\end{lem}
\begin{proof}
By using the identity
$$
    \div_q(gM^{-1}\gamma) = \gamma^\rT M^{-1}\d_q g + g \div_q (M^{-1}\gamma)
$$ along with
the condition (\ref{gpos}), it follows from (\ref{gdot}) that  the function $\ln g$ also satisfies a linear first-order PDE:
\begin{equation}
\label{gdot3}
  \d_t \ln g
  =
  \frac{1}{g}\d_t g
  =
  - \gamma^\rT M^{-1}\d_q \ln g - \div_q (M^{-1}\gamma).
\end{equation}
The time differentiation of (\ref{hfg}), combined with
(\ref{gdot3}), the FPKE (\ref{FPKE}) and the fact that the position PDF $g$ does not depend on the momentum variables,   leads to
\begin{align}
\nonumber
    \d_t h
    =&
    \frac{1}{g} \d_t f - h\d_t \ln g\\
\nonumber
    =&
    \frac{1}{g}
        \{H,gh\}
    +
    p^\rT M^{-1} F \d_ph
    +
    \frac{1}{2}
    \bra
        D, \d_p^2 h
    \ket
    +
    (
    \bra
        F, M^{-1}
    \ket
    -\d_t \ln g
    )
    h
\\
\nonumber
    =&
        \{H,h\}
    +
    p^\rT M^{-1} F \d_ph
    +
    \frac{1}{2}
    \bra
        D, \d_p^2 h
    \ket
    +
    (
    \bra
        F, M^{-1}
    \ket
    +
    \{H, \ln g\}
    -\d_t \ln g
    )
    h
\\
\nonumber
    =&
    \{H,h\}
    +
    p^\rT M^{-1} F \d_ph
    +
    \frac{1}{2}
    \bra
        D, \d_p^2 h
    \ket\\
\label{hdot1}
     & +
     (    \bra
        F, M^{-1}
    \ket
    +
    \{H,\ln g\} + \gamma^\rT M^{-1}\d_q \ln g + \div_q (M^{-1}\gamma))
     h.
\end{align}
 Here, use is also made of the derivation property $\{u,vw\} =\{u,v\}w+v\{u,w\}$ of the Poisson bracket (\ref{Poiss}), whereby 
  $\frac{1}{g}\{H, gh\} = \{H, h\} + \{H, \ln g\} h$ in view of the identity $ \frac{1}{g} \{H, g\} = \{H, \ln g\}$. Also,
\begin{equation}
\label{Hlng}
    \{H,\ln g\}
    =
    \d_q H^{\rT}\underbrace{\d_p \ln g}_{0} - \d_p H^{\rT}\d_q \ln g
     =
     -p^\rT M^{-1}\d_q \ln g,
\end{equation}
where the $p$-independence of the position PDF $g$ is used again. Now, substitution of (\ref{Hlng}) into (\ref{hdot1}) leads to
\begin{align}
\nonumber
    \d_t h
    =&
    \{H,h\}
    +
    p^\rT M^{-1} F \d_ph
    +
    \frac{1}{2}
    \bra
        D, \d_p^2 h
    \ket\\
\label{hdot2}
     & +
     (    \bra
        F, M^{-1}
    \ket
    -(p-\gamma)^\rT M^{-1}\d_q \ln g
    + \div_q (M^{-1}\gamma))
     h,
\end{align}
which,  in view of (\ref{varpi}), establishes (\ref{hdot}).
\end{proof}

For completeness,   we also note that by using (\ref{Poissdiv}) and representing the relevant terms of the PDE (\ref{hdot}) (or (\ref{hdot2})) in the divergence form, it follows that
\begin{align}
\nonumber
    \d_t h
    =&
    \overbrace{\div_p(h\d_q H) - \div_q(h M^{-1}p)}^{\{H,h\}}
    +
    \div_p
        \Big(
        h FM^{-1}p + \frac{1}{2} D\d_p h
    \Big)\\
\label{hdot3}
     & +
     (
     \div_q (M^{-1}\gamma)
    -\varpi^\rT M^{-1}\d_q \ln g
    )
     h,
\end{align}
where the structure of the right-hand side is in accordance with the preservation of the normalization property
\begin{equation}
\label{hnorm}
    \int_{\mR^n} h(p\mid q, t)\rd p = 1,
    \qquad
    q\in S,
\end{equation}
over the course of time $t\> 0$. Indeed,   the $\div_p(\cdot)$ terms in (\ref{hdot3}) do not contribute to the integral over the momentum variables, which, together with the $p$-independence of $M$, $\gamma$, $g$,  yields
\begin{align*}
    &\d_t \int_{\mR^n}
    h
    \rd p
    =
    \int_{\mR^n}
    (
    - \div_q(h M^{-1}p)
    +
     (
     \div_q (M^{-1}\gamma)
    -\varpi^\rT M^{-1}\d_q \ln g
    )h
    )\rd p\\
     =&
    -\div_q
    \Big(
    M^{-1}
    \underbrace{\int_{\mR^n}
    h p \rd p}_{\gamma}
    \Big)
    \!+\!
    \div_q (M^{-1}\gamma)
    \underbrace{\int_{\mR^n}
    h \rd p}_{1}
    \!-\!
    \underbrace{\int_{\mR^n}
    h\varpi\rd p}_{0}
    M^{-1}\d_q \ln g = 0,
\end{align*}
where the intermediate integrals (including  $\int_{\mR^n}
    h\varpi\rd p = \int_{\mR^n}
    hp\rd p - \gamma \int_{\mR^n}
    h\rd p = 0$) use (\ref{EPQ}), (\ref{varpi})  and (\ref{hnorm}).

The presence of the conditional momentum mean $\gamma$ (which depends linearly on $h$ according to (\ref{EPQ}))   as a factor on the right-hand side of (\ref{hdot})  makes the latter PDE an integro-differential equation with a quadratic nonlinearity with respect to $h$.  Furthermore, (\ref{hdot}) involves  the position PDF $g$ and is thus coupled to the PDE (\ref{gdot}), which,  in turn, is affected by $\gamma$ as mentioned before. This coupling is  illustrated in Figure~\ref{fig:couple}.
%==============================================================================
\begin{figure}[htbp]
\centering
\unitlength=1.1mm
\linethickness{0.4pt}
\begin{picture}(90.00,40.00)
    \put(20,31){\framebox(10,10)[cc]{(\ref{gdot})}}
    \put(25,31){\vector(0,-1){10}}
    \put(29,26){\makebox(0,0)[cc]{{\small$\d_t g$}}}

    \put(25,3){\makebox(0,0)[cc]{{\small$g$}}}
    \put(65,39){\makebox(0,0)[cc]{{\small$h$}}}
    \put(35,39){\makebox(0,0)[cc]{{\small$\gamma$}}}
    \put(20,11){\framebox(10,10)[cc]{\small$\int(\cdot)  \rd t$}}
    \put(25,11){\line(0,-1){5}}
    \put(15,6){\vector(1,0){45}}

    \put(15,6){\line(0,1){30}}
    \put(15,36){\vector(1,0){5}}
    \put(40,36){\vector(-1,0){10}}
    \put(40,31){\framebox(10,10)[cc]{\small$\int(\cdot)  p\rd p$}}
    \put(75,36){\vector(-1,0){25}}
    \put(65,36){\line(0,-1){5}}
    \put(60,21){\framebox(10,10)[cc]{\small$\int(\cdot) \rd t$}}
    \put(65,11){\vector(0,1){10}}
    \put(60,1){\framebox(10,10)[cc]{(\ref{hdot})}}
    \put(75,36){\line(0,-1){30}}
    \put(75,6){\vector(-1,0){5}}
    \put(61,16){\makebox(0,0)[cc]{{\small$\d_t h$}}}
\end{picture}
\caption{An informal block-diagram illustration of the coupling between the PDEs (\ref{gdot}),  (\ref{hdot}).
}
\label{fig:couple}
\end{figure}
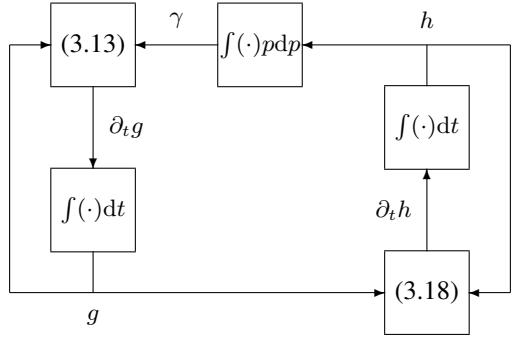

In view of the role which the  function $\gamma$  plays in the coupling of the PDEs (\ref{gdot}), (\ref{hdot}), we will also note the effect of symmetries in the conditional momentum PDF $h$  on $\gamma$. To this end, $h$ is decomposed as
\begin{equation}
\label{hhh}
    h = h_+ + h_-
\end{equation}
into functions $h_+$, $h_-$ which are, respectively,  symmetric and antisymmetric over the momentum variables:
\begin{equation*}
\label{h+-}
    h_\pm (p\mid q, t):= \frac{1}{2}(h (p\mid q, t)\pm h (-p\mid q, t)).
\end{equation*}
Since $\int_{\mR^n} h_-\rd p = 0$,  the symmetric part $h_+$ of (\ref{hhh}) inherits from $h$ the normalization property (\ref{hnorm}):
\begin{equation*}
\label{hnorm+}
    \int_{\mR^n} h_+(p\mid q, t)\rd p = 1,
    \qquad
    t\> 0,\ q\in S.
\end{equation*}
Similarly, $\int_{\mR^n} h_+ p\rd p = 0$, so that only the antisymmetric  part $h_-$ contributes to the conditional momentum mean (\ref{EPQ}):
\begin{equation}
\label{EPQ-}
  \gamma(t,q)
  =
  \int_{\mR^n}
   h_-(p\mid q,t) p
  \rd p.
\end{equation}
In view of the symmetry of the Hamiltonian $H$ in (\ref{HTV}) over the momentum variables, $\{H, h_+\}$ is antisymmetric, while $\{H, h_-\}$ is symmetric. Hence, (\ref{hdot}) can be represented in vector-matrix form as
\begin{equation}
\label{hdot+-}
    \d_t
    \begin{bmatrix}
      h_+\\
      h_-
    \end{bmatrix}
    =
        \begin{bmatrix}
          \cA & \cB\\
          \cB & \cA
        \end{bmatrix}
    \begin{bmatrix}
      h_+\\
      h_-
    \end{bmatrix},
\end{equation}
where $\cA$, $\cB$ are $(g,h)$-dependent integro-differential operators, acting on a function $\varphi\in C^2(S\x \mR^n, \mR)$ as
\begin{align}
\nonumber
    \cA(\varphi)
     :=&
    (
    \bra
        F, M^{-1}
    \ket
    +
    \gamma^\rT M^{-1}\d_q \ln g
    + \div_q (M^{-1}\gamma)
    )\varphi\\
\label{cA}
    & +
    p^\rT M^{-1} F \d_p\varphi
    +
    \frac{1}{2}
    \bra
        D, \d_p^2\varphi
    \ket, \\
\label{cB}
    \cB(\varphi)
    := &
    \{H,\varphi\}- p^\rT M^{-1}(\d_q \ln g)\varphi,
\end{align}
with $\gamma$ in (\ref{cA}) being linearly related to $h_-$ through (\ref{EPQ-}). The latter makes the right-hand side of (\ref{hdot+-}) depend bilinearly on $h_\pm$ in accordance  with the quadratic nonlinearity of the PDE (\ref{hdot}) whose right-hand side is representable as $\cA(h)+\cB(h)$. The structure of the $(2\x 2)$-matrix in (\ref{hdot+-}) reflects the fact that the operator
$\cA$ in (\ref{cA}) preserves the symmetry (or antisymmetry) of a function with respect to $p$, whereas $\cB$ in (\ref{cB})  changes it to the opposite.

\section{Statistical Mechanical Equilibrium and Invariant Measure}
\label{sec:inv}

According to the postulate of statistical mechanics \cite{R_1978}, the probability  distribution of the system state $X$ from (\ref{XQP})  in equilibrium with its environment, which is modelled as a heat bath at absolute temperature $\cT>0$,    minimises the Helmholtz free energy
\begin{align}
\nonumber
    \bF(f)
    & :=
    \bE H(X) - \cT  \bh(X)=
    \bE (H(X) + \cT  \ln f(X))\\
\label{bF}
    & =
    \int_{S\x \mR^n}
    f(x)
    (H(x) + \cT  \ln f(x))\rd x,
\end{align}
where the dependence of the random vector $X$ and its PDF $f$ on time is omitted for brevity.  The mean value $\bE H$ of the Hamiltonian (\ref{HTV}) pertains to the internal energy of the system, and
\begin{equation*}
\label{bh}
    \bh(X)
    :=
    -
    \bE \ln f(X)
    =
    -
    \int_{S\x \mR^n} f(x) \ln f(x)\rd x
\end{equation*}
is the differential entropy \cite{CT_2006} of $X$ (with respect to the Lebesgue measure $\lambda_{2n}$ on the phase space), where the standard continuity convention $0\ln 0 = 0$ is used.    The minimum value of (\ref{bF}) over the position-momentum PDF $f$ is achieved at  the unique Maxwell-Boltzmann PDF $f_*$, which is expressed in terms of the Hamiltonian $H$ as
\begin{equation}
\label{f*}
    f_*(x)
    =
    \frac{\re^{-\beta H(x)}}{Z(\beta)},
    \qquad
    x \in S\x \mR^n,
\end{equation}
where
\begin{equation}
\label{Sbeta}
    \beta
    :=
    \frac{1}{\cT }
\end{equation}
is the inverse temperature. Here, the units are chosen so that the Boltzmann constant (relating  the thermodynamic entropy to its information theoretic  counterpart \cite{CT_2006,ME_1981}) is $k_\rB = 1$,  and the temperature is measured in the units of energy. 
The normalization constant
\begin{equation}
\label{Z}
    Z(\beta)
    :=
    \int_{S \x \mR^n}
    \re^{-\beta H(x)}
    \rd x
\end{equation}
in (\ref{f*}) is the statistical mechanical partition function \cite{ME_1981}, which specifies the minimum free energy in (\ref{bF}) as
\begin{align}
\nonumber
    \bF(f_*)
    & =
    \min_{{\rm over\ PDFs}\ f\ {\rm on}\ S\x \mR^n} \bF(f) \\
\label{bF*}
    & =
    \bE_* (H + \cT  \ln f_*(X))=
    -\cT \ln Z(\beta),
\end{align}
where $\bE_*(\cdot)$ is the expectation over the equilibrium PDF  $f_*$.  For any PDF $f$,
the free energy (\ref{bF}) is linked to the relative entropy (or Kullback-Leibler informational divergence) \cite{CT_2006} of $f$ with respect to $f_*$:
\begin{align}
\nonumber
    \bD(f\| f_*)
    & :=
    \bE \ln \frac{f(X)}{f_*(X)}=
    \bE ( \beta H+\ln f(X))
    +
    \ln Z(\beta) \\
\label{bD}
    & =
    \beta
    (\bF(f)-\bF(f_*)),
\end{align}
as follows from   (\ref{f*}),  (\ref{bF*}). This also relates the variational property of $f_*$ as the minimiser for $\bF$ to the fact that the relative entropy $\bD(f\| f_*)$ is always nonnegative and vanishes only when $f=f_*$ almost everywhere \cite{CT_2006} in the sense that the set $\fX:=\{x\in S\x \mR^n: f(x)\ne f_*(x)\}$  satisfies $\int_\fX f(x)\rd x = 0$ (or equivalently, $\int_\fX f_*(x)\rd x = 0$).
Due to the structure (\ref{T})--(\ref{HTV})  of the Hamiltonian, the  equilibrium PDF (\ref{f*}) is factorised as
\begin{equation}
\label{fgh*}
    f_*(x)
    =
    g_*(q) h_*(p\mid q),
    \qquad
    x \in S\x \mR^n,
\end{equation}
where
\begin{align}
\nonumber
  g_*(q)
  & =
  \int_{\mR^n}
  f_*(q, p)
  \rd p
  =
  \frac{1}{Z(\beta)}
  \re^{-\beta V(q)}
  \int_{\mR^n}
  \re^{-\frac{1}{2}\beta \|p\|_{M(q)^{-1}}^2}
  \rd p\\
\label{g*}
  & =
  \frac{1}{Z(\beta)}
  (2\pi \cT )^{n/2}
  \re^{-\beta V(q)}
  \sqrt{\det M(q)},
  \qquad
  q \in S,
\end{align}
is the equilibrium position PDF, and
\begin{equation}
\label{h*}
  h_*(p\mid q)
  :=
  \frac{  (2\pi \cT )^{-n/2}}{\sqrt{\det M(q)}}
  \re^{-\frac{1}{2}\beta\|p\|_{M(q)^{-1}}^2},
  \qquad
  p \in \mR^n,
\end{equation}
is the equilibrium conditional momentum PDF. 
For any given position $q\in S$, the PDF $h_*(\cdot \mid q)$
specifies a Gaussian distribution in $\mR^n$ with the mean vector
\begin{equation}
\label{E*PQ}
    \bE_*
    (P\mid Q = q)
    :=
    \int_{\mR^n}
    p
    h_*(p\mid q)
    \rd p
    =0
\end{equation}
and the covariance matrix
\begin{equation}
\label{E*PPQ}
    \bE_*(PP^\rT \mid Q=q)
    :=
    \int_{\mR^n}
    pp^\rT
    h_*(p\mid q)
    \rd p
    =
    \cT  M(q).
\end{equation}
 Since the PDF (\ref{g*}) satisfies the normalization $\int_S g_*(q)\rd q = 1$, the partition function (\ref{Z}) takes the form
\begin{equation}
\label{Z1}
    Z(\beta)
    =
    (2\pi \cT )^{n/2}
    \int_S
    \re^{-\beta V(q)}
    \sqrt{\det M(q)}
    \rd q.
\end{equation}
If the potential energy $V$ achieves its global minimum value $\inf_{q\in S }V(q)$ at a unique  position  $q_* \in S$,
\begin{equation}
\label{q*}
    \Argmin_{q\in S}
    V(q)
     = \{q_*\},
\end{equation}
  with a positive definite Hessian matrix
\begin{equation*}
\label{V''}
    K
    :=
    V''(q_*)\succ 0
\end{equation*}
(the stiffness matrix at the equilibrium position),
then application of the Laplace method \cite[Theorem 4.1 on pp. 74, 75]{F_1977} to the integral in (\ref{Z1}) yields the following low-temperature asymptotics:
\begin{align}
\nonumber
    Z(\beta)
    & \sim
    (2\pi \cT )^{n/2}
    \re^{-\beta V(q_*)}
    \sqrt{\det M(q_*)}
    \int_{\mR^n}
    \re^{-\frac{1}{2}\beta \|q-q_*\|_K^2}
    \rd q\\
\label{Zasy}
    & =
    (2\pi \cT )^n
    \re^{-\beta V(q_*)}
    \sqrt{\frac{\det M(q_*)}{\det K}},
    \qquad
    {\rm as}\ \cT  \to 0+,
\end{align}
where use is made of an auxiliary Gaussian PDF
\begin{equation}
\label{ghat}
    \wh{g}(q)
    :=
    (2\pi \cT )^{-n/2} \sqrt{\det K}\, \re^{-\frac{1}{2}\beta \|q-q_*\|_K^2},
    \qquad
    q \in \mR^n,
\end{equation}
with the mean vector $q_*$ and the covariance matrix $\cT  K^{-1}$, which arises from the Taylor series approximation
\begin{equation*}
\label{Vquad}
    V(q) = V(q_*) + \frac{1}{2} \|q-q_*\|_K^2 + o(|q-q_*|^2),
      \qquad
      {\rm as}\  q \to q_*,
\end{equation*}
in view of the vanishing gradient $V'(q_*) = 0$. A combination of (\ref{Zasy}) with (\ref{g*}) clarifies the role of $\wh{g}$ in (\ref{ghat}) as a low-temperature  Gaussian approximation for the equilibrium  position PDF $g_*$ in a small neighbourhood of $q_*$ (as $\cT \to 0+$ and $q\to q_*$ in such a way that $|q-q_*| = O(\sqrt{\cT })$). The maxima of the equilibrium position PDF $g_*$ in (\ref{g*}) are the minima of   $\beta V -\frac{1}{2}\ln\det M$ in $S$,
\begin{equation}
\label{maxmin}
    \Argmax_{q\in S}
    g_*(q) =
    \Argmin_{q\in S}
    \Big(
        V(q) -\frac{\cT}{2}\ln\det M(q)
    \Big),
\end{equation}
and can differ from (\ref{q*}) due to the dependence  of the mass matrix $M$ on $q$.  However, as $\cT \to 0+$ (or, equivalently, $\beta \to +\infty$), the potential energy term $V$ in (\ref{maxmin}) becomes dominant, the $q$-dependence of $M$ loses its effect, and  the set of maxima of $g_*$  converges to $q_*$. In a particular case when the mass matrix $M$ is constant, the term $-\frac{\cT}{2}\ln\det M$ in (\ref{maxmin}) becomes irrelevant and the maxima of $g_*$ coincide with the minima of $V$ at any temperature $\cT >0$. The peaks of $g_*$ are increasingly pronounced for low values of $\cT $, which underlies  the use of Hamiltonian dynamics in the heavy-ball stochastic optimization algorithms for minimising the potential $V$ over $S$ in this case (see, for example, \cite{GP_2014} and references therein).

The PDF $f_*$  in (\ref{f*}), suggested by the postulate of equilibrium statistical mechanics,  corresponds to an invariant measure for the stochastic Hamiltonian system (\ref{dX}) if $f_*$ is a steady-state solution of the FPKE (\ref{FPKE}). A sufficient and necessary  condition for this property to hold, provided by the following lemma for completeness, is a multivariate version of the Einstein relation \cite[Eq. (3.14) on p. 260]{K_1966} between the damping and diffusion coefficients.

%%%%%%%%%%%%%%%%%%%%%%%%%%%%%%%%%%%%%%%%%%%%%%%%%%%%%%%%%%%%%%%%%%%%%%%%%%%%
\begin{lem}
\label{lem:FD}
The PDF (\ref{f*}) is invariant for the stochastic Hamiltonian system (\ref{dX}) if and only if 
the damping and diffusion matrices are related as
\begin{equation}
\label{FD}
    F(q) = \frac{1}{2}\beta D(q),
    \qquad
    q \in S.
\end{equation}
%\hfill$\square$
\end{lem}
\begin{proof}
Since $\{H, \varphi(H)\} = \varphi'(H)\{H,H\} = 0$ for any $\varphi \in C^1(\mR, \mR)$ in view of (\ref{Poiss}), (\ref{HH0}), so that the PDF (\ref{f*}) satisfies $\{H,f_*\} = 0$, the factorisation (\ref{fgh*}) reduces the stationarity condition
\begin{align}
\nonumber
    0  & =
    \cL^\dagger(f_*)
    =
    \{H,f_*\}
    +
    p^\rT M^{-1} F \d_pf_*
    +
    \bra
        F, M^{-1}
    \ket
    f_*
    +
    \frac{1}{2}
    \bra
        D, \d_p^2 f_*
    \ket\\
\label{cLf*}
     & =
    \Big(
    p^\rT M^{-1} F \d_p\ln f_*
    +
    \bra
        F, M^{-1}
    \ket
    +
    \frac{1}{2}
    (
    \bra
        D, \d_p^2 \ln f_*
    \ket
    +
    \|\d_p \ln f_*\|_D^2
    )
    \Big)
    f_*
\end{align}
for the FPKE (\ref{FPKE}) to
\begin{equation}
\label{FPKEh*}
    p^\rT M^{-1} F \d_p\ln h_*
    +
    \bra
        F, M^{-1}
    \ket
    +
    \frac{1}{2}
    (
    \bra
        D, \d_p^2 \ln h_*
    \ket
    +
    \|\d_p \ln h_*\|_D^2
    )
    =0.
\end{equation}
Here, use is made of the identities
\begin{equation}
\label{lnphi'}
    (\ln\phi)'
    =
    \phi'/\phi,
    \qquad
    (\ln\phi)''
    =
    \phi''/\phi
    -
    \phi'\phi'^\rT/\phi^2
\end{equation}
for positive functions $\phi \in C^2(\mR^r, \mR)$,  along with the relation $\ln f_* = \ln g_* + \ln h_*$ and the fact that the equilibrium position PDF $g_*$ does not depend on the momentum variables,  whereby $\d_p \ln f_* = \d_p \ln h_*$ and $\d_p^2 \ln f_* = \d_p^2 \ln h_*$.
It follows from (\ref{h*}) that the logarithmic gradient vector and the Hessian matrix of $h_*$ with respect to the momentum variables take the form
\begin{equation}
\label{hgradHess*}
    \d_p \ln h_*
    =
    -\beta M^{-1}p,
    \qquad
    \d_p^2 \ln h_*
    =
    -
    \beta M^{-1},
\end{equation}
whereby
(\ref{FPKEh*}) is equivalent to
\begin{align}
\nonumber
    0
     & =
    -\beta \|M^{-1}p\|_F^2
    +
    \bra F, M^{-1}\ket
    -
    \frac{1}{2}
    \beta
    \bra D, M^{-1}\ket
    +
    \frac{1}{2}
    \beta^2
    \|M^{-1}p\|_D^2\\
\label{FPKEh*1}
    & =
    \underbrace{
    \Bra
        F-\frac{1}{2} \beta D,
        M^{-1}
    \Ket
    }_{p{\rm-independent}}
    -
    \beta
    p^\rT
    \underbrace{
    M^{-1} \Big(F-\frac{1}{2} \beta D\Big) M^{-1}
}_{p{\rm-independent}}p.
\end{align}
The right-hand side of (\ref{FPKEh*1}) is a quadratic function of $p$ whose coefficients (as indicated)  depend only on $q$. Therefore,  the fulfillment of
(\ref{FPKEh*1})
for all $q\in S$ and $p \in \mR^n$
is equivalent to (\ref{FD}). Alternatively, this can also be obtained by considering $h_*$ as a steady-state solution of the PDE (\ref{hdot}).
\end{proof}

Since the relation (\ref{FD}) secures correspondence between the stochastic dynamics (\ref{dX}) and the  statistical mechanical equilibrium   postulate, this condition is assumed to be satisfied in what follows. In particular, due to (\ref{FD}), the energy balance relation (\ref{EHdot}) takes the form
\begin{equation}
\label{EHdotD}
    (\bE H)^{^\centerdot}
    =
    \frac{1}{2}
    \bE
    \bra
        M^{-1} D,
        I_n-\beta\bE(PP^{\rT}\mid Q)M^{-1}
    \ket.
\end{equation}
Although the diffusion matrix $D$, which  parameterises the damping matrix $F$ through  (\ref{FD}), does not enter the Hamiltonian $H$ or the equilibrium PDF $f_*$ in (\ref{f*}), it affects the dynamics of the system and its  convergence to equilibrium. In addition to (\ref{EHdotD}),   this influence manifests itself in terms of other quantities, including  the free energy (or relative entropy) functionals.

\section{Position-Momentum Relative Entropy Decomposition}
\label{sec:posmoment}

For the sake of brevity, we denote the relative entropy (\ref{bD}),   as a function of time,  by
\begin{equation}
\label{cF}
  \cF(t)
   :=
  \bD(f(t,\cdot)\| f_*)
  =
  \bE\theta(t,X(t))
  =
  \int_{S\x \mR^n}
  f(t,x)
  \theta(t,x)
  \rd x,
  \quad
  t \> 0.
\end{equation}
Here, use is made of an auxiliary  function $\theta: \mR_+ \x S \x \mR^n\to \mR$ along with related functions $\xi: \mR_+ \x S\to \mR$ and $\eta: \mR_+ \x S\x\mR^n \to \mR$, which are defined as logarithmic PDF ratios
\begin{align}
\label{theta}
  \theta(t,q,p)
  & :=
  \ln \frac{f(t,q,p)}{f_*(q,p)} = \xi(t,q) + \eta(t,q,p),\\
\label{xi}
  \xi(t,q)
  & :=
  \ln \frac{g(t,q)}{g_*(q)},\\
\label{eta}
  \eta(t,q,p)
  & :=
  \ln \frac{h(p\mid q, t)}{h_*(p\mid q)},
\end{align}
and, in addition to (\ref{gpos}),  the conditional momentum PDF $h$ is also assumed to be positive everywhere:
\begin{equation}
\label{hpos}
  h(p\mid q, t)>0,
  \qquad
  t \> 0,\
  q \in S,\
  p \in \mR^n.
\end{equation}
The second equality in (\ref{theta}) follows from the factorizations (\ref{fgh}), (\ref{fgh*}) of the PDF $f$ and its invariant counterpart $f_*$. For any time $t\> 0$,  the decomposition in (\ref{theta}) in terms of (\ref{xi}), (\ref{eta}) allows (\ref{cF}) to be split as
\begin{equation}
\label{bDfgh}
    \cF(t)
     =
    \bE
    (
        \xi(t,Q(t))
        +
        \eta(t,Q(t),P(t))
    )
    =
    \cG(t) + \cH(t)
\end{equation}
into the corresponding relative entropies for the position PDF $g$  and the conditional momentum PDF $h$:
\begin{align}
\label{cG}
    \cG(t)
    & := \bD(g(t,\cdot)\| g_*)
     =
     \bE \xi(t,Q(t))
     =
     \int_S
     g(t,q)
     \xi(t,q)
     \rd q,\\
\nonumber
    \cH(t) & := \bD(h(\cdot\mid \cdot,t )\| h_*)
    =
    \bE \eta(t,Q(t),P(t))
    =
     \int_{S\x \mR^n}
     f(t,q,p)
     \eta(t,q,p)
     \rd q
     \rd p\\
\label{cH}
     & =
     \int_{S}
     g(t,q)
     \Big(
     \underbrace{
     \int_{\mR^n}
     h(p\mid q, t)
     \eta(t,q,p)
     \rd p}_{\bD(h(\cdot  \mid q, t) \| h_*(\cdot \mid q))}
     \Big)
     \rd q.
\end{align}
The innermost integral on the right-hand side of (\ref{cH}) is the relative entropy of the conditional momentum PDF $h(\cdot \mid q, t)$ with respect to $h_*(\cdot \mid q)$ for a given $q \in S$. For ease of reference, Table~\ref{tab:ass} summarises  the above quantities along with  their associations.
\begin{table}[h!]
\centering
\caption{An informal association for the system variables, energy functions, PDFs, logarithmic PDF ratios and entropies.\label{tab:ass}}
\begin{tabular}{|l|l|c|c|c|}
 \hline
variable	& energy function	& PDF & logarithmic PDF ratio & entropy\\
 \hline
state $X$		& Hamiltonian $H$			& $f$ & $\theta$ & $\cF$    \\
position $Q$		& potential $V$			& $g$ & $\xi$    & $\cG$\\
momentum $P$		& kinetic $T$			& $h$ & $\eta$   & $\cH$\\
\hline
\end{tabular}
\end{table}

The relative entropies provide upper bounds for the $L^1$-deviations  of the corresponding PDFs from their invariant counterparts. More precisely,
application of \index{Pinsker's inequality} Pinsker's inequality (see, for example, \cite[Lemma 2.5 on p. 88]{T_2009}) to (\ref{cF}) leads to an upper bound for the total variation distance \cite{Z_1983} between the probability distributions with the PDFs $f$ and $f_*$:
\begin{align}
\nonumber
    d(f(t,\cdot),f_*)
    & :=
    \sup_{B \in \fB(S\x \mR^n)}
    \Big|
    \int_B
    (f(t,x)-f_*(x))
    \rd x
    \Big|\\
\label{Pinsker}
    & =
    \frac{1}{2}
    \|f(t,\cdot)-f_*\|_1
    \<
    \sqrt{\frac{1}{2} \cF(t)},
\end{align}
where the supremum is over the $\sigma$-algebra of Borel subsets of the phase space $S\x \mR^n$, and
$
    \|\varphi\|_1
    :=
    \int_{S\x \mR^n}
    |\varphi(x)|
    \rd x
$
is the $L^1$-norm of an absolutely integrable function $\varphi: S\x\mR^n\to \mR$. A similar inequality holds for (\ref{cG}):
\begin{equation}
\label{ggL1}
    \|g(t,\cdot)-g_*\|_1
    :=
    \int_{S}
    |g(t,q)-g_*(q)|
    \rd q
    \<
    \sqrt{2 \cG(t)}.
\end{equation}
In application to (\ref{cH}), a combination of Pinsker's and Jensen's inequalities  leads to a weighted $L^1$-bound:
\begin{align}
\nonumber
    \sqrt{2\cH(t)}
     & =
     \sqrt{
     2\int_{S}
     g(t,q)
     \bD(h(\cdot  \mid q, t) \| h_*(\cdot \mid q))
     \rd q}\\
\nonumber
    & \>
     \sqrt{\int_{S}
     g(t,q)
     \|h(\cdot  \mid q, t)-h_*(\cdot \mid q)\|_1^2
     \rd q}\\
     \label{rhoL1}
    & \>
     \int_{S}
     g(t,q)
     \|h(\cdot  \mid q, t)-h_*(\cdot \mid q)\|_1
     \rd q
     \>
          \|
     \varrho(t,\cdot) -\wh{\varrho}(t,\cdot)
     \|_1
\end{align}
in view of the convexity of the function $\mR \ni u\mapsto u^2$ and the $L^1(\mR^n, \mR)$-norm $\|\cdot\|_1$. Here,
\begin{equation}
\label{momPDF}
    \varrho(t,p)
    :=
    \int_S
    f(t,q,p)\rd q
    =
    \int_S
    g(t,q)
    h(p\mid q, t)
    \rd q
\end{equation}
is the actual momentum PDF, and
\begin{equation}
\label{momPDFhat}
    \wh{\varrho}(t,p)
    =
    \int_S
    g(t,q)
    h_*(p\mid q)
    \rd q,
    \qquad
    t \> 0,\
    p \in \mR^n
\end{equation}
is the PDF which the momentum $P(t)$ would have if its conditional PDF $h$ coincided with $h_*$. Since the total variation distance $d$ between probability measures does not exceed $1$,  the inequalities (\ref{Pinsker})--(\ref{rhoL1}) are useful only for small values   of the relative entropy ($< 2$).

\section{Relative Entropy Dissipation}
\label{sec:relent}

Regardless of the particular structure of the Markovian dynamics (\ref{SH1}), (\ref{SH2}),  the relative entropy (\ref{cF}) with respect to the invariant measure is nonincreasing  in time (which is interpreted as the second law of thermodynamics \cite{CT_2006} in application to stochastic systems). However, the specific form of the generator $\cL$ in  (\ref{cL}) (the sparsity of the diffusion matrix $        \small\begin{bmatrix}
          0 & 0 \\
          0 & D
        \end{bmatrix}$ of the state process $X$) influences the relative entropy dissipation as discussed below.

\begin{lem}
\label{lem:bDdot}
For the stochastic Hamiltonian system (\ref{SH1}), (\ref{SH2})  satisfying the dam\-ping-diffusion condition  (\ref{FD}) along with (\ref{gpos}), (\ref{hpos}), the relative entropy (\ref{cF}) of the position-momentum  PDF $f$ with respect to the invariant PDF $f_*$ in (\ref{f*}) evolves as
\begin{equation}
\label{bDdot}
    \dot{\cF}
     =
    -\frac{1}{2}
    \bE
    (
    \|\d_p \eta\|_D^2
    )
    \< 0,
\end{equation}
where $\eta$ is the logarithmic PDF ratio (\ref{eta}) associated with the conditional momentum PDF $h$ and its invariant counterpart $h_*$ in (\ref{h*}). %\hfill$\square$
\end{lem}
\begin{proof}
The time differentiation of (\ref{cF}) can be carried out by applying the generator $\cL$ from (\ref{cL}) over the position and momentum variables of the logarithmic PDF ratio $\theta$ in (\ref{theta}) as
\begin{align}
\nonumber
    \dot{\cF}
    & =
    (
    \bE
    \theta(t,X(t))
    )^{^\centerdot}
    =
    \bE
    (
    \d_t \theta
    +
    \cL(\theta)
    )\\
\nonumber
    & =
    \int_{S\x \mR^n}
    (
    \d_t f
    +
    f
    \cL(\theta)
    )
    \rd x
    =
    \int_{S\x \mR^n}
    \Big(
    f_*
    \cL\Big(\frac{f}{f_*}\Big)
    -\frac{1}{2}
    f
    \|\d_p \theta\|_D^2
    \Big)
    \rd x\\
\label{bDffdot}
    & =
    \int_{S\x \mR^n}
    \Big(
    \underbrace{\cL^\dagger(f_*)}_{0}
    \frac{f}{f_*}
    -\frac{1}{2}
    f
    \|\d_p \theta\|_D^2
    \Big)
    \rd x
    =
    -\frac{1}{2}
    \bE
    (
    \|\d_p \theta\|_D^2
    ),
\end{align}
where the right-hand side has the structure of Dirichlet forms \cite{FOT_2011} for Markov processes.  
Here, the relation $f\d_t \theta = f\d_t \ln f = \d_t f$ is obtained from (\ref{theta})  and the time independence of the invariant PDF $f_*$. Also,
the identity $\int_{S \x \mR^n} \d_t f\rd x = \d_t \int_{S \x \mR^n} f\rd x = 0$ follows from the normalization property $\int_{S \x \mR^n} f\rd x = 1$ of the PDF $f$, and (\ref{cLf*}) is used for the equilibrium PDF (\ref{f*}) under the condition (\ref{FD}).  Furthermore, use is made of the Fleming logarithmic transformation \cite{F_1982}
(see also \cite[Eq.~(81) on p.~201]{BFP_2002}), which, in application to the generator
$\cL$ in (\ref{cL}) and a twice continuously differentiable function $S\x \mR^n \ni (q,p) \mapsto \varphi(q,p) > 0$,  leads to
\begin{equation}
\label{logtrans}
    \cL(\ln \varphi)
     =
     \frac{1}{\varphi}
    \cL(\varphi)
    -
    \frac{1}{2}
    \|\d_p\ln \varphi\|_D^2,
\end{equation}
where the quadratic form is influenced by the position-momentum structure (\ref{xqp}) of the  system variables and the absence of diffusion in the position ODE (\ref{SH1}).
In (\ref{bDffdot}), the relation (\ref{logtrans}) is used for the PDF ratio $\varphi:= \frac{f}{f_*}$, whose logarithm  $\theta = \ln \varphi$ is decomposed into the sum in the second equality in (\ref{theta}) under the conditions (\ref{gpos}), (\ref{hpos}). 
Since $\xi$ on the right-hand side of (\ref{theta}) is independent of $p$, it does not  contribute to $\d_p \theta =\d_p \eta $, whereby (\ref{bDffdot}) establishes (\ref{bDdot}).
\end{proof}

In view of the dissipation inequality (\ref{bDdot}), the relative entropy $\cF$ does not increase in time and can be used as a Lyapunov functional in the context of convergence of the state PDF  $f$ to $f_*$. However, the right-hand side of (\ref{bDdot}) is negative only when $h$ deviates from $h_*$.
The monotonicity of the ``total'' relative entropy $\cF$ in (\ref{bDfgh}) does not extend to its components $\cG$, $\cH$ in (\ref{cG}), (\ref{cH}).  As the following lemma shows, the time derivative of $\cG$ is organised as a bilinear (rather than a definite quadratic) form in $\xi$, $h$.

\begin{lem}
\label{lem:cGdot}
Under the assumptions of Lemma~\ref{lem:bDdot}, the position entropy (\ref{cG}) satisfies
\begin{align}
\label{cGdot}
  \dot{\cG}
  & =
  \bE (\gamma^\rT  M^{-1}\d_q\xi ),\\
\label{cGddot}
      \ddot{\cG}
  & =
  \bE
  \Big(
    \d_q\xi^\rT M^{-1} P\d_t \eta -  \gamma^\rT  M^{-1}\d_q
    \Big( \frac{1}{g_*}\div_q(g_*M^{-1}\gamma)\Big) \Big),
\end{align}
where $\xi$ is the corresponding logarithmic PDF ratio in (\ref{xi}), and $\gamma$ is the conditional momentum  mean from (\ref{EPQ}).  %\hfill$\square$
\end{lem}
\begin{proof}
Similarly to the proofs of Lemmas~\ref{lem:gdot} and \ref{lem:bDdot},
the time derivative of (\ref{cG}) can be computed by using (\ref{SH1}) as 
\begin{align}
\nonumber
    \dot{\cG}
    & =
    (
    \bE
    \xi(t,Q(t))
    )^{^\centerdot}
    =
    \bE
    (
    \d_t \xi
    +
    \d_q\xi^\rT \dot{Q}
    ) =
    \bE
    (
    \d_t \xi
    +
    \d_q\xi^\rT M^{-1}
    \bE(P\mid Q)
    )\\
\label{bDggdot}
    & =
    \int_S
    (
    \d_t g
    +
    g
    \d_q\xi^\rT M^{-1}    \gamma
    )
    \rd q
    =
    \int_S
    g
    \d_q\xi^\rT M^{-1}    \gamma
    \rd q,
\end{align}
which establishes (\ref{cGdot}). The second last equality in (\ref{bDggdot}) employs
the relation $g\d_t \xi = g\d_t \ln g = \d_t g$ obtained from (\ref{xi})  and the time independence of the invariant position PDF $g_*$. This is combined with
$\int_S \d_t g\rd q = \d_t \int_S g\rd q = 0$ due to the normalization property
\begin{equation}
\label{gnorm}
     \int_S
    g(t,q)\rd q = 1
\end{equation}
of the PDF $g$ for any $t\> 0$.
Now, by differentiating (\ref{cGdot}) in time and using the time independence of the mass matrix $M$ (whereby $\d_t M = 0$), it follows, similarly to (\ref{bDggdot}),   that
\begin{align}
\nonumber
  \ddot{\cG}
  & =
  \bE (\d_t(\gamma^\rT  M^{-1}\d_q\xi) + \dot{Q}^\rT\d_q(\gamma^\rT  M^{-1}\d_q\xi)  )\\
\nonumber
  & =
  \bE (
    \d_t\gamma^\rT  M^{-1}\d_q\xi
    +
    \gamma^\rT  M^{-1}\d_t\d_q\xi
    +
    \gamma^\rT  M^{-1}\d_q(\gamma^\rT  M^{-1}\d_q\xi)  )\\
\label{cGddot1}
  & =
  \bE (
    \d_t\gamma^\rT  M^{-1}\d_q\xi
    +
    \gamma^\rT  M^{-1}\d_q(\d_t \xi+ \gamma^\rT  M^{-1}\d_q\xi)  ),
\end{align}
where $\d_t\d_q\xi = \d_q\d_t\xi  $  due to the interchangeability of differentiation over  time and the position variables. Also, since $\d_t h = h\d_t \ln h = h \d_t \eta$ by the time independence of the invariant conditional momentum PDF $h_*$ in (\ref{eta}), the time differentiation of (\ref{EPQ})  yields
\begin{equation}
\label{gammadot}
    \d_t \gamma
    =
    \int_{\mR^n}
    p\d_t h \rd p
    =
    \int_{\mR^n}
    p h \d_t \eta \rd p
    =
    \bE (P \d_t \eta \mid Q=q).
\end{equation}
By combining (\ref{gdot3}) with the identity $\ln g = \xi + \ln g_*$ in view of (\ref{xi}), it follows that
$$
\d_t \xi
  =
  \d_t \ln g
  =
  - \gamma^\rT M^{-1}\d_q (\xi + \ln g_*) - \div_q (M^{-1}\gamma),
$$
and hence,
\begin{align}
\nonumber
    \d_t \xi+ \gamma^\rT  M^{-1}\d_q\xi
    & =
    - \gamma^\rT M^{-1}\d_q \ln g_* - \div_q (M^{-1}\gamma)\\
\label{xixi}
    & =
    - \frac{1}{g_*}\div_q(g_*M^{-1}\gamma).
\end{align}
Substitution of (\ref{gammadot}), (\ref{xixi}) into the right-hand side of (\ref{cGddot1}) establishes (\ref{cGddot}).
\end{proof}

While the monotonicity of the position-momentum entropy $\cF$ in time is not shared by its components $\cG$, $\cH$ in (\ref{bDfgh}), it makes them  influence each other. Indeed, since
\begin{equation}
\label{FGHdot}
    \dot{\cG} + \dot{\cH} = \dot{\cF} \< 0,
\end{equation}
the  time derivatives $\dot{\cG}$, $\dot{\cH}$ cannot be simultaneously positive,
so that an increase in one of the entropies $\cG$, $\cH$ causes a decrease in the other as if there were an entropy flow between them.  
This ``waterbed'' effect is particularly prominent in the vicinity of those moments of time when the relative entropy dissipation rate in (\ref{bDdot}) vanishes.
The set
\begin{equation}
\label{fS}
    \fS:= \{t>0:\ \dot{\cF}(t)=0\}
\end{equation}
(more precisely, its scattered structure) 
plays an important role in the convergence of the system to equilibrium.  The local behaviour of the entropies $\cF$, $\cG$, $\cH$ and the PDF $h$ at such instants is discussed below.

\begin{lem}
\label{lem:fS}
Under the conditions of Lemma~\ref{lem:bDdot}, at any entropy dissipation break time described by (\ref{fS}),  the relative entropies $\cF$, $\cG$, $\cH$ in (\ref{cF}), (\ref{cG}), (\ref{cH}) satisfy
\begin{align}
\label{FFsign}
    \ddot{\cF}
    & = 0,
    \qquad
    \dddot{\cF}\< 0,\\
\label{GGsign}
    \dot{\cG}
    & = 0,
    \qquad
    \, \ddot{\cG}=
    -\ddot{\cH} \<  0, \\
\label{HHsign}
    \cH & = 0,
    \qquad
    \dot{\cH} = 0,
    \qquad
    \ddot{\cH} \> 0,
    \qquad
    t \in \fS.
\end{align}
%\hfill$\square$
\end{lem}
\begin{proof}
Since $\dot{\cF}(t)=0$ at any $t\in \fS$, then the function $\dot{\cF}$, which is nonpositive everywhere  in view of (\ref{bDdot}),   achieves its global maximum  value $0$ at any such moment of time. Hence, the first two derivatives of $\dot{\cF}$ satisfy (\ref{FFsign}) as necessary conditions of the maximum. In view of (\ref{gpos}), (\ref{bDdot}), for any $t\> 0$,  the property $\dot{\cF}(t)=0$ holds if and only if $h(\cdot \mid \cdot, t) = h_*$ everywhere in the phase space $S\x \mR^n$. Due to (\ref{cH}), this allows the set (\ref{fS}) to be represented as
\begin{equation}
\label{fS1}
    \fS= \{t>0:\ \cH(t)=0\}.
\end{equation}
Therefore, at any $t \in \fS$, the nonnegative function $\cH$ achieves its global minimum value $0$, whereby its first two derivatives satisfy (\ref{HHsign}) as necessary conditions of the minimum. For any $t \in \fS$, the inequality in (\ref{FGHdot}) becomes  an equality,  and hence, $\dot{\cG} = - \dot{\cH} = 0$ in view of the second equality in (\ref{HHsign}). By a similar reasoning, the relation $\ddot{\cG} = \ddot{\cF}- \ddot{\cH}$, which follows from (\ref{bDfgh}), leads to $\ddot{\cG}(t) = -\ddot{\cH}(t)\< 0$ in (\ref{GGsign}) for any $t\in \fS$ due to the first equality in (\ref{FFsign}).
\end{proof}

As mentioned in the proof of Lemma~\ref{lem:fS} in regard to (\ref{fS1}), the conditional momentum PDF $h(\cdot \mid \cdot, t)$ coincides with its invariant counterpart $h_*$ at every $t \in \fS$ from (\ref{fS}) (and hence, the momentum PDF $\varrho(t,\cdot)$ in (\ref{momPDF}) coincides with $\wh{\varrho}(t,\cdot)$ in (\ref{momPDFhat}) everywhere in $\mR^n$ at any such time $t$, which can also be obtained by combining  (\ref{rhoL1}) with (\ref{fS1})).  Therefore, for any $t\in \fS$, the conditional momentum mean (\ref{EPQ}) vanishes everywhere (that is, $\gamma(t,q)=0$ for all $q\in S$)  in accordance with (\ref{E*PQ}) and hence, so also does $\d_t g$ for the position PDF $g$ in view of (\ref{gdot}); see Figure~\ref{fig:gh}. \begin{figure}[htbp]
{\centering
\includegraphics[width=7 cm]{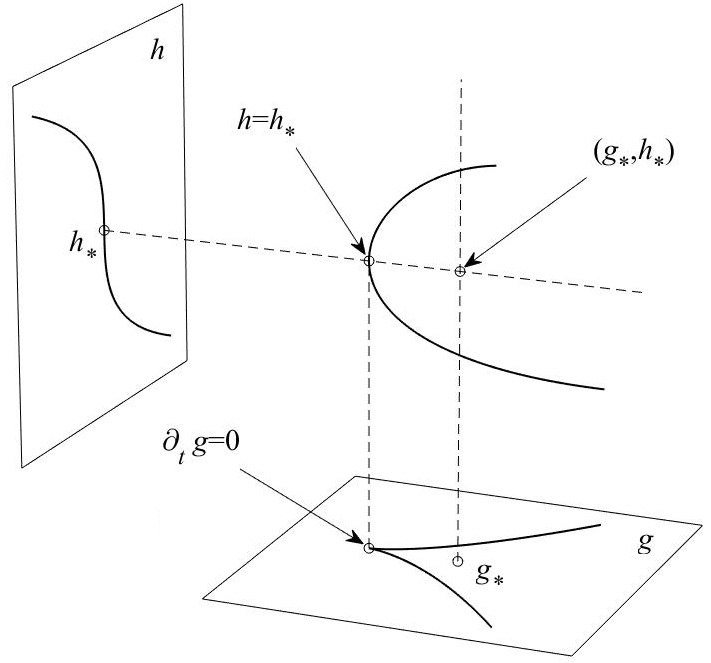}
\caption{An illustration of the local behaviour of the pair of the position and conditional momentum PDFs $g$, $h$ (as a trajectory in an infinite dimensional $(g,h)$-manifold) in a neighbourhood of an entropy dissipation break time (when $h=h_*$). At any such moment of time, $\d_t g = 0$.
\label{fig:gh}}}
\end{figure}
Moreover, for any $t\in \fS$, not only $\dot{\cF} = 0$, but also $(\bE H)^{^\centerdot} = 0$ in view of (\ref{EHdotD}) since a combination of (\ref{Sbeta}) with (\ref{E*PPQ}) for the PDF $h_*$  in (\ref{h*}) implies that $I_n-\beta\bE_*(PP^{\rT}\mid Q)M^{-1} = 0$.  At the same time, as the following lemma shows, $h$ keeps evolving unless $g= g_*$.

\begin{lem}
\label{lem:hdot*}
Under the conditions of Lemma~\ref{lem:bDdot}, at any entropy dissipation break time in (\ref{fS}),  
\begin{equation}
\label{lnhdot*}
  \d_t \eta
  =
  -p^\rT M^{-1} \d_q \xi,
  \qquad
  t \in \fS,\
  q \in S,\
  p \in \mR^n,
\end{equation}
where $\xi$, $\eta$ are the logarithmic PDF ratios (\ref{xi}), (\ref{eta}) associated with the position and conditional momentum PDFs $g$, $h$  and their invariant counterparts $g_*$, $h_*$  in (\ref{g*}), (\ref{h*}). %\hfill$\square$
\end{lem}
\begin{proof}
Under the condition (\ref{gpos}) on the position PDF $g$, for  any conditional momentum PDF $h>0$,  by dividing both parts of the third equality in (\ref{hdot1}) by $h$ and using (\ref{Hlng}), it follows that the logarithmic time derivative of $h$ takes the form
\begin{align}
\nonumber
    \d_t \ln h
    = &
        \{H,\ln h\}
    +
    p^\rT M^{-1} F \d_p\ln h
    +
    \frac{1}{2}
    \bra
        D, \d_p^2 \ln h
    \ket
    +
    \frac{1}{2}
    \|\d_p \ln h\|_D^2\\
\label{lnhdot}
    & +
    \bra
        F, M^{-1}
    \ket
    -p^\rT M^{-1}\d_q \ln g
    -\d_t \ln g,
\end{align}
where use is also made of  (\ref{lnphi'}).  
We will now evaluate the right-hand side of (\ref{lnhdot}) at the invariant conditionally Gaussian momentum PDF $h(\cdot \mid \cdot, t)=h_*$ from (\ref{h*}), which corresponds to the case when $t\in \fS$.
The logarithmic gradient vector of $h_*$ with respect to the position variables is computed as
\begin{equation}
\label{dlnhdq}
    \d_q \ln h_*
     =
    \frac{1}{2}
    (
    \beta
    p^\rT M_k p -
    \bra
        M^{-1},
        \d_{q_k} M
    \ket)_{1\< k \< n},
\end{equation}
where the $\mS_n$-valued maps $M_1, \ldots, M_n$ from (\ref{muk}) are used. 
Also, since $h_*$ in (\ref{h*}) has zero conditional momentum mean (\ref{E*PQ}), so that  $\gamma=0$ in (\ref{EPQ}), then, in view of (\ref{gdot3}),
\begin{equation}
\label{gdot0}
  \d_t \ln g = 0,
  \qquad
  t \in \fS,\
  q \in S.
\end{equation}
By substituting (\ref{hgradHess*}), (\ref{dlnhdq}), (\ref{gdot0}) into the right-hand side of (\ref{lnhdot}) and using the damping-diffusion condition (\ref{FD}), it follows that
\begin{align}
\nonumber
    \d_t \ln h
    =&
    \d_q H^{\rT}\d_p \ln h_* - \d_p H^{\rT}\d_q \ln h_*\\
\nonumber
    & +
    p^\rT M^{-1} F \d_p\ln h_*
    +
    \frac{1}{2}
    \bra
        D, \d_p^2 \ln h_*
    \ket
    +
    \frac{1}{2}
    \|\d_p \ln h_*\|_D^2
    \\
\nonumber
     & +
         \bra
        F, M^{-1}
    \ket
    -p^\rT M^{-1}\d_q \ln g,\\
\nonumber
    =&
    -\beta  \d_q H^\rT M^{-1}p
    -
    \frac{1}{2}
    p^\rT M^{-1}
    (
    \beta
    p^\rT M_k p -
    \bra
        M^{-1},
        \d_{q_k} M
    \ket)_{1\< k \< n}\\
\nonumber
    &
    -\beta
    \|M^{-1}p\|_F^2
    -
    \frac{1}{2}\beta
    \bra
        D, M^{-1}
    \ket
    +
    \frac{1}{2}\beta^2
    \|M^{-1}p\|_D^2
    \\
\nonumber
     & +
         \bra
        F, M^{-1}
    \ket
    -p^\rT M^{-1}\d_q \ln g\\
\nonumber
    =&
    -\beta  \Big(    V'
    -
    \frac{1}{2}
    (
        p^{\rT}M_k p
    )_{1\< k\< n}
    \Big)^\rT M^{-1}p\\
\nonumber
    & -
    \frac{1}{2}
    p^\rT M^{-1}
    (
    \beta
    p^\rT M_k p -
    \bra
        M^{-1},
        \d_{q_k} M
    \ket)_{1\< k \< n}
    -p^\rT M^{-1}\d_q \ln g\\
\nonumber
    = &
    p^\rT M^{-1}
    \Big(
        \frac{1}{2}
    (\bra
        M^{-1},
        \d_{q_k} M
    \ket)_{1\< k \< n}    -\beta V'
            -\d_q \ln g
    \Big)\\
\label{Gausslnhdot}
    = &
    -p^\rT M^{-1}
    \d_q\xi,
\end{align}
which establishes (\ref{lnhdot*}) since $\d_t \eta=\d_t \ln h$ by the time independence of the invariant conditional momentum PDF $h_*$ in (\ref{eta}).  The last equality in (\ref{Gausslnhdot}) is obtained by combining the relation
\begin{equation*}
\label{dqlng*}
    \d_q \ln g_*
    =
        \frac{1}{2}
    (\bra
        M^{-1},
        \d_{q_k} M
    \ket)_{1\< k \< n}    -\beta V'
\end{equation*}
for the invariant position PDF $g_*$ in (\ref{g*}) with $ \xi = \ln g - \ln g_*$ which follows from (\ref{xi}).
\end{proof}

The relation (\ref{lnhdot*}) shows that for any time $t \in \fS$ (that is, when $h=h_*$),  the time derivative of the conditional momentum PDF $h$ vanishes everywhere in the phase space $S\x \mR^n$ if and only if $g=g_*$ everywhere in $S$, in which case the system state has the equilibrium PDF $f=f_*$.    The following lemma enhances the inequalities in (\ref{FFsign})--(\ref{HHsign}).

\begin{lem}
\label{lem:F...}
Under the conditions of Lemma~\ref{lem:bDdot}, at any entropy dissipation break time in (\ref{fS}),  the third and second-order  time derivatives of the entropies $\cF$, $\cH$  in (\ref{cF}), (\ref{cH}) satisfy
\begin{align}
\label{F...}
  \dddot{\cF}
  & =
  -
  \bE
  (
    \|
        M^{-1}\d_q \xi
    \|_D^2
  ),\\
\label{H..}
  \ddot{\cH}
  & =
  \cT
  \bE
  (
    \|
        \d_q \xi
    \|_{M^{-1}}^2
  ),
  \qquad
  t \in \fS,
\end{align}
where $\xi$ is the logarithmic PDF ratio (\ref{xi}) for the position PDF $g$ and its invariant counterpart $g_*$ in (\ref{g*}). %\hfill$\square$
\end{lem}
\begin{proof}
By twice differentiating the expectation in (\ref{bDdot}) and letting $h(\cdot\mid \cdot, t)=h_*$, which holds at any $t \in \fS$, it follows that
\begin{align}
\nonumber
    -2\dddot{\cF}
     & =
     (\bE
    (
    \|\d_p \eta\|_D^2
    ))^{^{\centerdot\centerdot}}\\
\nonumber
     & =
     \int_{S\x \mR^n}
     (
     \d_t^2f \|\d_p \eta\|_D^2
     +
     2\d_tf \d_t(\|\d_p \eta\|_D^2)
     +
     f \d_t^2(\|\d_p \eta\|_D^2)
     )
     \rd x\\
\label{-2Fddd}
    & =
    2
     \int_{S\x \mR^n}
     f \|\d_t\d_p \eta\|_D^2
     \rd x
     = 2
     \bE
     (
     \|\d_t\d_p \eta\|_D^2).
\end{align}
Indeed, for any such time $t$, (\ref{eta}), (\ref{cH}), (\ref{fS1}) imply that $\eta(t,\cdot, \cdot)=0$ everywhere in the phase space $S \x \mR^n$, and hence, $\d_p \eta = 0$.  Therefore,  in view of the time independence of the diffusion matrix $D$,
\begin{align*}
    \d_t(\|\d_p \eta\|_D^2)
    & =
    2 \d_p \eta^\rT D \d_t\d_p \eta=0,\\
    \d_t^2(\|\d_p \eta\|_D^2)
    & =
    2 (\|
        \d_t\d_p \eta\|_D^2 + \d_p \eta^\rT D \d_t^2\d_p \eta)
        =
        2 \|
        \d_t\d_p \eta\|_D^2.
\end{align*}
The interchangeability of differentiation in time and the momentum variables  yields
\begin{equation}
\label{dtdp}
    \d_t\d_p \eta
     = \d_p  \d_t\eta
     =
    -\d_p(p^\rT M^{-1} \d_q \xi)
    =
    -M^{-1} \d_q \xi,
    \qquad
    t \in \fS.
\end{equation}
Here, use is also made of the relation (\ref{lnhdot*}) together with the $p$-independence of the mass matrix $M$ and the logarithmic PDF ratio $\xi$ in (\ref{xi}). Substitution of (\ref{dtdp}) into (\ref{-2Fddd}) leads to (\ref{F...}). Furthermore, by substituting (\ref{lnhdot*}) into (\ref{cGddot}) and recalling that at any time $t \in \fS$,   (\ref{EPQ}) satisfies $\gamma =0$ in view of $h=h_*$ and  (\ref{E*PQ}), it follows that
\begin{align}
\nonumber
\ddot{\cG}
   & =
   -
  \bE
  (
    \d_q\xi^\rT M^{-1} PP^\rT M^{-1} \d_q \xi)\\
\label{G..}
   & =
   -
  \bE
  (
    \d_q\xi^\rT M^{-1} \bE_*(PP^\rT \mid Q) M^{-1} \d_q \xi)    =
   -
   \cT
  \bE
  (
    \|\d_q\xi\|_{M^{-1}}^2),
\end{align}
where use is also made of (\ref{E*PPQ}).  The relation (\ref{H..}) can now be obtained from (\ref{G..}) and the second equality in (\ref{GGsign}).
\end{proof}

 Using Lemmas~\ref{lem:bDdot}--\ref{lem:F...},
the following theorem establishes a strict monotonicity property for the position-momentum entropy (\ref{cF}).

\begin{thm}
\label{th:entmono}
For the stochastic Hamiltonian system (\ref{SH1}), (\ref{SH2})  satisfying the damping-diffusion condition  (\ref{FD}),  the relative entropy (\ref{cF}) of the position-mo\-men\-tum  PDF $f$ with respect to the invariant PDF $f_*$ in (\ref{f*}) strictly decreases in time until $f_*$ is reached. %\hfill$\square$
\end{thm}
\begin{proof}
Denote by
\begin{equation}
\label{tau}
    \tau
    :=
    \inf
    \{
        t \> 0:\
        f(t,\cdot) = f_*\
        {\rm everywhere\ in}\
        S\x\mR^n
    \}
\end{equation}
the first time when the position-momentum PDF $f$ reaches the invariant PDF $f_*$, with the convention that $\tau:= +\infty$ if this never happens (that is, if the set in (\ref{tau}) is empty). By appropriately restricting the set $\fS$ in  (\ref{fS}) to
\begin{equation}
\label{fStau}
    \fS_\tau:
    =
    (0,\tau)\bigcap \fS
    =
    \{0< t< \tau:\ \dot{\cF}(t)=0\},
\end{equation}
it follows that  for any $t \in \fS_\tau$, the position PDF is different from its invariant counterpart: $g(t,\cdot)\ne g_*$.  Indeed, together with $h(\cdot \mid \cdot, t)=h_*$, which holds at any $t \in \fS$,      the fulfillment of $g(t,\cdot)=g_*$ would imply that $f(t,\cdot)=f_*$ (that is, the equilibrium is already reached) and hence,  $t \> \tau$, thus contradicting the inequality $t< \tau$.  Therefore, in view of (\ref{F...}),
\begin{equation}
\label{F...neg}
    \dddot{\cF} < 0,
    \qquad
    t \in \fS_\tau.
\end{equation}
In combination with (\ref{fS}) and the first equality in (\ref{FFsign}), (\ref{F...neg}) implies that any such instant is a stationary point of inflection for the relative entropy $\cF$ as a function of time. In a small neighbourhood of such a point, $\cF$ behaves asymptotically as a strictly decreasing cubic parabola:
\begin{equation}
\label{cubic}
    \cF(s)=\cF(t) + \frac{1}{6}\dddot{\cF}(t)(s-t)^3 + o(|s-t|^3),
    \qquad
    {\rm as}\
    s \to t \in \fS_\tau.
\end{equation}
Therefore, the set $\fS_\tau$ in (\ref{fStau}) consists of isolated points and is countable (hence, of zero Lebesgue measure), which, in view of
\begin{equation*}
\label{Ffot...neg}
    \dot{\cF} < 0,
    \qquad
    t \in (0,\tau)\setminus \fS,
\end{equation*}
makes $\dot{\cF}$ strictly negative almost everywhere in the interval $[0,\tau]$. The latter implies that $\cF(t)-\cF(s) = \int_s^t \dot{\cF}(u)\rd u < 0$ for all $0\< s < t \< \tau$, and hence, $\cF$ is a strictly decreasing function of time over $[0,\tau]$.
\end{proof}

The strict monotonicity of the relative entropy $\cF$, proved in Theorem~\ref{th:entmono},  employs the observation that the position-momentum distribution does not ``stay'' at those pairs $(g,h)$,  where $\dot{\cF}=0$,  unless the system has reached the equilibrium PDF $f_*$.
 As illustrated in Figure~\ref{fig:FGH},
\begin{figure}[htbp]
{\centering
\includegraphics[width=7 cm]{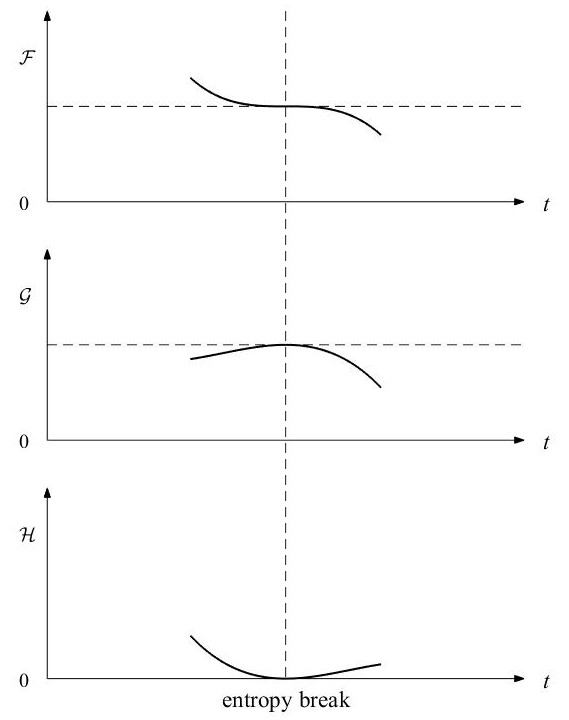}
\caption{An illustration of the local behaviour of the entropies $\cF$, $\cG$, $\cH$ from (\ref{cF}), (\ref{cG}), (\ref{cH}) in the vicinity of a pre-equilibrium  entropy dissipation break time in (\ref{fStau}),  which is described asymptotically by the cubic parabola (\ref{cubic}) and two quadratic parabolas, whose nonzero coefficients (\ref{F...}), (\ref{H..}) make the inequalities in
(\ref{FFsign})--(\ref{HHsign}) strict.
\label{fig:FGH}}}
\end{figure}
the local behaviour of $\cF$ and its components $\cG$, $\cH$ in the vicinity of any pre-equilibrium   entropy dissipation  break time in (\ref{fStau}) is described asymptotically by the cubic parabola (\ref{cubic}) and two (concave and convex)  quadratic parabolas, respectively.  The inequalities in (\ref{FFsign})--(\ref{HHsign}) for the coefficients of these parabolas are strict  in view of (\ref{F...}), (\ref{H..}), which can be interpreted as a local exchange between the position and conditional momentum  entropies $\cG$, $\cH$, so that the increase in one of these entropy components is compensated by the decrease in the other at the level of the first and second-order terms of their Taylor
series expansions. This ``waterbed''  behaviour is similar to the total energy dissipation and potential-kinetic energy exchange relations (\ref{H...})--(\ref{T0}) in the case of deterministic dissipative Hamiltonian dynamics and can be regarded as a manifestation of the BKL principle \cite{BK_1952,L_1960} (mentioned in the Introduction) in application to the position-momentum entropy $\cF$ as a Lyapunov functional for the FPKE in the stochastic  Hamiltonian setting.

\section{Linearised Dynamics of Logarithmic PDF Ratios}
\label{sec:lin}

The pair
\begin{equation}
\label{zeta}
    \zeta:= (\xi, \eta) = \Theta(f)
\end{equation}
of the logarithmic PDF ratios (\ref{xi}), (\ref{eta}) results from a nonlinear bijective  transformation $\Theta$ of the position-momentum PDF $f$. The latter is uniquely recovered from $\zeta$ as $f = \Theta^{-1}(\zeta) =  f_*\re^{\xi+\eta}$ in view of (\ref{theta}).
The transformation $\Theta$ and its inverse $\Theta^{-1}$ allow the FPKE (\ref{fdot}) to be represented in terms of $\zeta$ as
\begin{equation}
\label{zetadot}
    \d_t \zeta
    =
    \Theta'(f)(\d_t f)
    =
    \Theta'(f)(\cL^\dagger(f))
    =
    \Theta'(f)(\cL^\dagger(\Theta^{-1}(\zeta))),
\end{equation}
where the linear operator $\Theta'(f)$ is the formal Frechet derivative of $\Theta$ with respect to  $f$. In accordance with the position and conditional momentum PDF dynamics (\ref{gdot}), (\ref{hdot}) of the stochastic Hamiltonian system, (\ref{zetadot}) consists of two equations for the evolution of $\xi$, $\eta$:
\begin{align}
\label{xidot}
    \d_t \xi
     = &
    - \frac{1}{g_*}\div_q(g_*M^{-1}\gamma) - \gamma^\rT  M^{-1}\d_q\xi,\\
\nonumber
    \d_t \eta
    = &
        \{H,\eta + \ln h_*\}
    +
    p^\rT M^{-1} F \d_p(\eta + \ln h_*)\\
\nonumber
    & +
    \frac{1}{2}
    \bra
        D, \d_p^2 (\eta + \ln h_*)
    \ket
    +
    \frac{1}{2}
    \|\d_p (\eta + \ln h_*)\|_D^2\\
\nonumber
    & +
    \bra
        F, M^{-1}
    \ket
    -p^\rT M^{-1}\d_q (\xi + \ln g_*)
    -\d_t \xi\\
\nonumber
    = &
        \{H,\eta\}
    -\frac{1}{2}\beta
    p^\rT M^{-1} D \d_p\eta
    +
    \frac{1}{2}
    \bra
        D, \d_p^2 \eta
    \ket
    +
    \frac{1}{2}
    \|\d_p \eta\|_D^2
    -p^\rT M^{-1}\d_q \xi
    -\d_t \xi\\
\nonumber
    =&
    \Big(
        \d_q H -
        \frac{1}{2}
        \beta
        D M^{-1}p
    \Big)^\rT \d_p \eta
    -
    p^\rT M^{-1}\d_q(\xi+\eta)\\
\label{etadot}
    & +
    \frac{1}{2}
    (\bra
        D,
        \d_p^2 \eta
    \ket
    +
    \|\d_p \eta\|_D^2) - \d_t\xi.
\end{align}
Here, (\ref{xidot}) follows directly from (\ref{xixi}), while (\ref{etadot}) is obtained by substituting $\ln h = \eta + \ln h_*$ into (\ref{lnhdot}) and using (\ref{gdot3}) along with the damping-diffusion condition (\ref{FD}) and (\ref{EPQ}) in the form
\begin{align}
\label{EPQeta}
  \gamma(t,q)
   =
  \int_{\mR^n}
   p h_*(p\mid q) \re^{\eta(t,q,p)}
  \rd p,
  \qquad
  t\> 0,\
  q \in S.
\end{align}
In (\ref{etadot}),  use is also made of the fact that the invariant position and conditional momentum PDFs $g_*$, $h_*$ satisfy
\begin{align*}
        \{H,\ln h_*\}
     & +
    p^\rT M^{-1} F \d_p\ln h_*\\
    & +
    \frac{1}{2}
    (\bra
        D, \d_p^2 \ln h_*
    \ket
    +
    \|\d_p \ln h_*\|_D^2)
    +
    \bra
        F, M^{-1}
    \ket
    \underbrace{-p^\rT M^{-1}\d_q \ln g_*}_{\{H, \ln g_*\}} = 0
\end{align*}
in view of the stationarity condition (\ref{FPKEh*}) and the relation $\{H, \ln g_*\} + \{H, \ln h_*\} = \{H, \ln f_*\} = 0$.

For the system near the equilibrium,  when the PDFs $g$, $h$ are close to their invariant counterparts $g_*$, $h_*$, so that the logarithmic PDF ratios $\xi$, $\eta$ in (\ref{xi}), (\ref{eta}) are ``close'' to $0$ together with relevant derivatives (a precise meaning of this proximity needs a separate consideration),  the equations (\ref{xidot})--(\ref{EPQeta}) admit the following linearisation:
\begin{align}
\label{xidotlin}
  \d_t \xi
  \approx&
      - \frac{1}{g_*}\div_q(g_*M^{-1}\gamma), \\
\label{etadotlin}
    \d_t \eta
    \approx &
    \Big(
        \d_q H -
        \frac{1}{2}
        \beta
        D M^{-1}p
    \Big)^\rT \d_p \eta
    -
    p^\rT M^{-1}\d_q(\xi+\eta)
    +
    \frac{1}{2}
    \bra
        D,
        \d_p^2 \eta
    \ket
    - \d_t\xi,\\
\label{EPQetalin}
  \gamma
   \approx &
  \int_{\mR^n}
   p h_*\eta
  \rd p.
\end{align}
The linearised form of the normalization conditions (\ref{gnorm}), (\ref{hnorm}) on $g$, $h$, which is preserved by (\ref{xidotlin}), (\ref{etadotlin}), is given by
\begin{align}
\label{normlin1}
    \int_S g_* \xi \rd q
    & \approx
    \int_S (\re^\xi-1) g_* \rd q
    =
    \int_S(g-g_*) \rd q
    =
    0,\\
\label{normlin2}
    \int_{\mR^n}
    h_* \eta \rd p
    & \approx
    \int_{\mR^n}
    (\re^\eta -1)
    h_*
    \rd p
    =
    \int_{\mR^n}
    (h-h_*)
    \rd p
    =
    0.
\end{align}
The linearisation in (\ref{xidotlin})--(\ref{normlin2})  neglects the terms $    - \gamma^\rT M^{-1}\d_q \xi$ and
$
    \frac{1}{2}
    \|\d_p \eta\|_D^2
$
of the second order of smallness with respect to scaling the functions $\xi$, $\eta$ in (\ref{xidot}), (\ref{etadot}),  along with
\begin{equation}
\label{2nd}
    \re^z-1-z = \frac{1}{2}z^2 + O(z^3),
    \qquad
    {\rm as}\
    z\to 0.
\end{equation}
By reorganising the pair (\ref{zeta}) into an $\mR^2$-valued function,  the linearised dynamics of the logarithmic PDF ratios in (\ref{xidotlin})--(\ref{EPQetalin})   are represented in vector-matrix form as
\begin{equation}
\label{xietadot}
    \d_t \zeta
    \approx
    \Lambda (\zeta),
    \qquad
    \zeta
    :=
    \begin{bmatrix}
      \xi\\
      \eta
    \end{bmatrix},
    \qquad
    \Lambda
    :=
    \begin{bmatrix}
          0 & \Xi\\
          \Phi & \Psi
    \end{bmatrix},
\end{equation}
where $\Phi$, $\Psi$, $\Xi$  are linear operators,  acting on functions $\varphi\in  C^1(S, \mC)$ and $\psi\in C^2(S\x \mR^n, \mC)$ as
\begin{align}
\label{Phi}
    \Phi(\varphi)
    & :=
        -
    p^\rT M^{-1}\d_q\varphi,\\
\label{Psi}
    \Psi(\psi)
    & :=
    \Phi(\psi)
    +
    \Big(
        \d_q H -
        \frac{1}{2}
        \beta
        D M^{-1}p
    \Big)^\rT \d_p \psi
    +
    \frac{1}{2}
    \bra
        D,
        \d_p^2 \psi
    \ket  - \Xi(\psi),\\
\label{Xi}
    \Xi(\psi)
    & :=
    -\frac{1}{g_*}
    \div_q
    \Big(
        g_*
        M^{-1}
        \int_{\mR^n}
        p h_*
        \psi
        \rd p
    \Big).
\end{align}
The integro-differential operator $\Xi$ originates from substituting the right-hand side of (\ref{EPQetalin}) for $\gamma$ into (\ref{xidotlin}). In (\ref{xietadot}), the operators $\Phi$, $\Psi$, $\Xi$ act over the position and momentum variables of $\xi$, $\eta$. Any eigenvalue $\lambda \in \mC$ of the linear operator $\Lambda$ in  (\ref{xietadot}) and the corresponding eigenfunctions $\varphi$, $\psi$ subject to the linear constraints (\ref{normlin1}), (\ref{normlin2}) satisfy
\begin{equation}
\label{eig}
    \lambda \varphi = \Xi (\psi),
    \qquad
    \lambda^2 \psi
    =
    (
        \Phi \circ \Xi  + \lambda \Psi
    )(\psi).
\end{equation}
The composition $\Phi \circ \Xi$ of the operators $\Phi$, $\Xi$ in  (\ref{Phi}), (\ref{Xi}), which  affects the spectrum of $\Lambda$ through  (\ref{eig}), acts as
\begin{equation}
\label{PhiXipsi}
    \Phi(\Xi(\psi))
    =
    p^\rT M^{-1}
    \d_q
    \Big(
    \frac{1}{g_*}
    \div_q
    \Big(
        g_*
        M^{-1}
        \int_{\mR^n}
        p h_*
        \psi
        \rd p
    \Big)
    \Big)
\end{equation}
and is self-adjoint with respect to the inner product $\bra u, v\ket_{f_*}   := \int_{S\x \mR^n} \overline{u}v f_* \rd x$  in the weighted Hilbert space $L_{f_*}^2(S\x \mR^n, \mC)$ (with $\overline{(\cdot)}$ the complex conjugate).  Indeed, by applying the Gauss-Ostrogradsky theorem and integrating by parts, it follows that
\begin{align}
\nonumber
    \bra
        \Phi(\varphi),
        \psi
    \ket_{f_*}
    & =
    -
    \int_{S \x \mR^n}
    p^\rT M^{-1}
    \d_q \overline{\varphi}
    \psi
    f_*
    \rd x    \\
\nonumber
    & =
    -
    \int_S
    \d_q \overline{\varphi}^\rT
    M^{-1}
    \Big(
    \int_{\mR^n}
    ph_* \psi
    \rd p
    \Big)
    g_*
    \rd q    \\
\nonumber
    & =
    \int_S
    \overline{\varphi}
    \div_q
    \Big(
    g_*
    M^{-1}
    \int_{\mR^n}
    ph_* \psi
    \rd p
    \Big)
    \rd q    \\
\label{braket*}
    & =
    -
    \int_S
    \overline{\varphi}
    \Xi(\psi)
    g_*
    \rd q
     =
     -
    \bra
        \varphi,
        \Xi(\psi)
    \ket_{g_*}
\end{align}
for any functions $\varphi \in C^1(S, \mC)$ and  $\psi \in C^1(S \x \mR^n, \mC)$ of bounded support,
where $\bra \mu, \nu\ket_{g_*}:= \int_S \overline{\mu} \nu g_* \rd q$ is a similar inner product in the weighted Hilbert space $L_{g_*}^2(S,\mC)$. The relation (\ref{braket*}) implies that $\Phi$, $\Xi$, as operators acting between $L_{g_*}^2(S, \mC)$ and $L_{f_*}^2(S\x \mR^n, \mC)$, satisfy
\begin{equation}
\label{PhiXidagger}
    \Phi = -\Xi^\dagger ,
\end{equation}
whereby $\Phi\circ \Xi$ in (\ref{PhiXipsi}) is a negative semi-definite self-adjoint operator on the appropriate subspace of $L_{f_*}^2(S\x \mR^n, \mC)$:
\begin{equation}
\label{PhiXineg}
    \Phi\circ \Xi = -\Xi^\dagger\circ \Xi\preccurlyeq 0.
\end{equation}
The negative semi-definiteness  is also seen directly from (\ref{braket*}) with $\varphi:= \Xi(\psi)$, which yields
\begin{equation}
\label{dir}
        \bra
        \Phi(\Xi(\psi)),
        \psi
    \ket_{f_*}
    = -
    \|\Xi(\psi)\|_{g_*}^2\< 0,
\end{equation}
where $\|\cdot\|_{g_*}$ is the weighted norm in $L_{g_*}^2(S, \mC)$.

On the other hand, since the invariant PDF $f_*$ satisfies $\Theta(f_*) =  0$ or, equivalently, $\Theta^{-1}(0) = f_*$, then (\ref{xidotlin})--(\ref{EPQetalin}) can be regarded as the linearisation  of (\ref{zetadot}) near the equilibrium $\zeta =0$. Therefore, the operator $\Lambda$ in (\ref{xietadot}), which, in view of (\ref{PhiXidagger}), takes the form
\begin{equation}
\label{LamXi}
    \Lambda
    =
    \begin{bmatrix}
          0 & \Xi\\
          -\Xi^\dagger & \Psi
    \end{bmatrix},
\end{equation}
is related by a similarity transformation
\begin{equation}
\label{Lambda}
    \Lambda
    =
    \Theta'(f_*)\circ \cL^\dagger \circ (\Theta^{-1})'(0)
\end{equation}
(and hence, is isospectral) to the operator $\cL^\dagger$ in (\ref{fdot}) since
$(\Theta^{-1})'(0) = (\Theta'(f_*))^{-1}$. The representation (\ref{LamXi}) shows that,  in the absence of $\Psi$, the operator $\Lambda$   would be skew self-adjoint on $L_{g_*}^2(S, \mC)\x L_{f_*}^2(S\x \mR^n, \mC)$, with its eigenvalues in (\ref{eig}) being   purely imaginary:  $\lambda = \pm i\omega$, where the eigenfrequencies $\omega\> 0$ are the singular values of $\Xi$. Therefore, it is due to the operator $\Psi$ in (\ref{Psi}) that $\Lambda$ (and hence, $\cL^\dagger$ in (\ref{Lambda})) acquires a spectrum with negative real parts, securing an exponentially fast convergence to the equilibrium. Note that every eigenvalue $\lambda$ in (\ref{eig}) is a root of a quadratic equation
\begin{equation}
\label{quad}
    \|\psi\|_{f_*}^2 \lambda^2
    -
    \bra \psi, \Psi(\psi)\ket_{f_*} \lambda
    +
    \|\Xi(\psi)\|_{g_*}^2
    =
    0,
\end{equation}
whose coefficients depend on the corresponding eigenfunction $\psi$. This equation is
obtained by applying the inner product $\bra \psi, \cdot\ket_{f_*}$ to  both sides of the second equality in (\ref{eig}) and using the weighted norm $\|\cdot\|_{f_*}$ in $L_{f_*}^2(S\x \mR^n, \mC)$ along with (\ref{PhiXineg}), (\ref{dir}). While the leading coefficient and the free term in (\ref{quad}) satisfy $\|\psi\|_{f_*}^2 >0$ and $\|\Xi(\psi)\|_{g_*}^2\> 0$,
the second coefficient $    -
    \bra \psi, \Psi(\psi)\ket_{f_*} $ can, in general,  take complex values.
The above clarifies the role of the position-momentum conditioning for spectral analysis of the PDF dynamics in the dissipative stochastic Hamiltonian setting.

For completeness, we also note that near the equilibrium  $\zeta = 0$ (where the linearisation (\ref{xietadot}) applies), the position and conditional momentum entropies $\cG$, $\cH$ in (\ref{cG}), (\ref{cH}) are of the second order of smallness with respect to $\xi$, $\eta$, and the quadratic term in (\ref{2nd}) has to be taken into account as a correction to (\ref{normlin1}), (\ref{normlin2}). More precisely,
\begin{align}
\nonumber
    \cG
    & =
    \int_S
    g_*\re^{\xi}\xi \rd q
    =
    \int_S
    \Big(\re^\xi-1+\frac{1}{2}\xi^2\Big)
    g_*\rd q +    \wt{\cG}    \\
\label{cGquad}
    & =
    \underbrace{\int_S (g-g_*)\rd q}_0
    +
    \frac{1}{2}
    \int_S g_* \xi^2\rd q
    +
    \wt{\cG}
    =
    \frac{1}{2}
    \int_S g_* \xi^2\rd q
    +
    \wt{\cG},
\end{align}
with
\begin{equation}
\label{cGrem}
    \wt{\cG} :=
    \frac{1}{6}
    \int_S
    g_*
    \re^{\mu \xi}
    (\mu\xi + 2)
    \xi^3
    \rd q,
\end{equation}
where the function $\mu: S \to (0,1)$ originates from the Lagrange remainder for $\phi(z):= \re^z (z-1)+1 -\frac{1}{2} z^2
$ with $\phi(0) = \phi'(0)=\phi''(0)=0$ and $\phi'''(z) = \re^z(z+2)$.
By a similar reasoning, the position-momentum entropy (\ref{cF}) admits the representation
\begin{align}
\nonumber
    \cF
    & =
    \int_{S\x \mR^n}
    f_*\re^{\theta}\theta \rd x
    =
    \int_{S\x \mR^n}
    \Big(\re^\theta-1+\frac{1}{2}\theta^2\Big)
    f_*\rd x +    \wt{\cF}    \\
\nonumber
    & =
    \underbrace{\int_{S\x \mR^n} (f-f_*)\rd x}_0
    +
    \frac{1}{2}
    \int_{S\x \mR^n} f_* \theta^2\rd x
    +
    \wt{\cF}\\
\label{cFquad}
    & =
    \frac{1}{2}
    \Big(
    \int_S
    g_* \xi^2\rd q
    +
    \int_{S\x \mR^n} f_* \eta^2\rd x
    \Big)
    +
    \int_{S\x \mR^n}
    f_* \xi\eta\rd x
    +
    \wt{\cF},
\end{align}
where use is also made of the logarithmic PDF ratio $\theta$ from (\ref{theta}) along with a remainder term
\begin{equation}
\label{cFrem}
    \wt{\cF} :=
    \frac{1}{6}
    \int_{S\x \mR^n}
    f_*
    \re^{\sigma \theta}
    (\sigma\theta + 2)
    \theta^3
    \rd x
\end{equation}
and an auxiliary function $\sigma: S \x \mR^n \to (0,1)$.  In turn, the last integral in (\ref{cFquad}) takes the form
\begin{align}
\nonumber
    \int_{S\x \mR^n}
    f_* \xi\eta\rd x
    & =
    \int_S
    \Big(
    \int_{\mR^n}
    h_* \eta
    \rd p
    \Big)
    g_* \xi
    \rd q\\
\nonumber
    & =
    \int_S
    \Big(
    \int_{\mR^n}
    \Big(
        \re^\eta - 1 - \frac{1}{2}\re^{\nu\eta}\eta^2
    \Big)
    h_*
    \rd p
    \Big)
    g_* \xi
    \rd q    \\
\nonumber
    & =
    \int_S
    \Big(
    \underbrace{\int_{\mR^n}
    (
        h-h_*
    )
    \rd p}_0
    \Big)
    g_* \xi
    \rd q
     -
     \frac{1}{2}
     \int_{S\x \mR^n}
     f_*
     \xi\re^{\nu\eta}\eta^2
     \rd x\\
\label{int}
     & =
     -
     \frac{1}{2}
     \int_{S\x \mR^n}
     f_*
     \xi\re^{\nu\eta}\eta^2
     \rd x
\end{align}
in view of (\ref{normlin2}), with $\nu: S\x \mR^n \to (0,1)$ an auxiliary function.  By  combining (\ref{bDfgh}) with (\ref{cGquad}), (\ref{cFquad}), (\ref{int}), the conditional momentum entropy $\cH$ in (\ref{cH}) is represented as
\begin{equation}
\label{cHquad}
    \cH
    =
    \cF - \cG
    =
        \frac{1}{2}
    \int_{S\x \mR^n} f_* \eta^2\rd x
    +
    \wt{\cH},
\end{equation}
where
\begin{equation}
\label{cHrem}
    \wt{\cH}
    :=
     \wt{\cF} - \wt{\cG}
     -
     \frac{1}{2}
     \int_{S\x \mR^n}
     f_*
     \xi\re^{\nu\eta}\eta^2
     \rd x.
\end{equation}
The remainder terms $\wt{G}$, $\wt{\cF}$, $\wt{\cH}$ in (\ref{cGrem}), (\ref{cFrem}), (\ref{cHrem}) are of the third order of smallness with respect to $\xi$, $\eta$ and  admit bounds using the family of Renyi relative entropies \cite{R_1961}.  In the framework of the linearised dynamics (\ref{xietadot}),
the quadratic approximation
$$
    \cF
    \approx
    \frac{1}{2}
    (
        \|\xi\|_{g_*}^2
        +
        \|\eta\|_{f_*}^2
    )
$$
of the position-momentum entropy
near the equilibrium,
resulting from (\ref{cFquad}) together with
$$
    \cG
    \approx
    \frac{1}{2}
    \|\xi\|_{g_*}^2,
    \qquad
    \cH
    \approx
    \frac{1}{2}
    \|\eta\|_{f_*}^2
$$
from (\ref{cGquad}), (\ref{cHquad}), is
an infinite-dimensional analogue of quadratic Lyapunov functions for finite-dimensional linear systems \cite{AM_1990}.

\section{Conclusion}\label{sec:conc}

For the class of multivariable stochastic Hamiltonian systems,  governed by the position ODE and the  momentum SDE with conservative, Langevin viscous  damping and external  random forces,  we have employed the factorisation of the joint position-moment PDF into the marginal position PDF and the conditional momentum PDF  given the position.  This position-momentum conditioning has lead to a decomposition of the FPKE into two coupled PDEs and higher-order dissipation relations for the position and conditional momentum Kullback-Leibler relative entropies with respect to the  corresponding Maxwell-Boltzmann equilibrium PDFs under the multivariate Einstein damping-diffusion condition. We have shown that all those entropy dissipation break instants, when  the position-momentum entropy has zero time derivative before the equilibrium is reached, are isolated  strong local maxima and minima of the position and conditional momentum entropies, respectively, which compensate each other in such a way that  the position-momentum entropy has stationary inflection points at those moments of time. An analogy has been discussed between the BKL principle scenarios  in application to the position-momentum entropy and the Hamiltonian as Lyapunov functionals in the stochastic and deterministic settings. We have also considered the linearised dynamics of the logarithmic PDF ratios near the equilibrium and outlined the use of the position-momentum conditioning for spectral analysis of the PDF dynamics in the stochastic Hamiltonian setting.

\end{document}